\documentclass[a4paper,12pt,fleqn]{article}
\usepackage{xcolor}
\usepackage{graphicx}
\usepackage{subcaption}
 \usepackage[all]{xy}
\usepackage[edges]{forest}
\usetikzlibrary{shapes.geometric,arrows.meta}
\usepackage{dsfont}
\usepackage[toc]{appendix}
\usepackage[applemac]{inputenc}
\usepackage{amssymb,amsfonts,amscd,amsmath,mathrsfs,amsthm}
\usepackage{bbm}
\usepackage{sidecap}
\usepackage{parskip}
\usepackage{hyperref}
\hypersetup{
     colorlinks=true,
     linkcolor=blue,
     filecolor=blue,
     citecolor = blue,      
     urlcolor=blue,
     }
\usepackage[a4paper,top=2cm,bottom=2cm,left=2cm,right=2cm,marginparwidth=1.75cm]{geometry}
\AfterEndEnvironment{figure}{\noindent\ignorespaces}
\newtheorem{df}{Definition}[section]
\newtheorem{thm}{Theorem}[section]
\newtheorem{prop}[thm]{Proposition}
\newtheorem{cor}[thm]{Corollary}
\newtheorem{lem}[thm]{Lemma}

\usepackage{graphicx} 
\DeclareMathOperator*{\argmax}{argmax} 
\DeclareMathOperator*{\argmin}{argmin} 

\def\E{{\mathbb E}}    
\def\N{{\mathbb N}}    
\def\P{{\mathbb P}}    
\def\R{{\mathbb R}}    

\begin{document}
\title{The geometry of moral decision making}
\author{R. Friedrich\footnote{E-mail: roland.friedrich@uzh.ch}}
\maketitle
\begin{abstract}
We show how (resource) bounded rationality can be understood as the interplay of two fundamental moral principles: deontology and utilitarianism. 
In particular, we interpret deontology as a regularisation function in an optimal control problem, coupled with a free parameter, the inverse temperature, to shield the individual from expected utility. 
We discuss the information geometry of bounded rationality and aspects of its relation to rate distortion theory. A central role is played by Markov kernels and regular conditional probability, which are also studied geometrically. A gradient equation is used to determine the utility expansion path.
Finally, the framework is applied to the analysis of a disutility model of the restriction of constitutional rights that we derive from legal doctrine. The methods discussed here are also relevant to the theory of autonomous agents.
\end{abstract}
\tableofcontents
\section{Introduction}
Consider this: a trolley is bound to hit a group of people, but a bystander can intervene and redirect the trolley so that it hits only one person on a separate track, as Welzel considered in his essay, ``Zum Notstandsproblem"~\cite{welzel1951notstandsproblem}. The question of how a bystander should decide in this situation has been extensively and controversially debated, with two fundamental moral theories clashing: utilitarianism and deontology.

Deontological theories, in the Kantian tradition, postulate that there are ethical rules, called norms, that specify that a given action is always right or wrong in a given situation, regardless of the consequences, cf. e.g.~\cite{skurski1983new,neumann1994moral}. In game theory, norms are defined as exogenous rules that select actions and indicate how players should behave, cf. e.g.~\cite{lopez2006introducing}.
In utilitarianism, as postulated by Mill and Bentham, and as opposed to deontology, the optimal rule to follow in a given situation is the one that produces the greatest utility, see e.g.~\cite{casebeer2003moral,skurski1983new,neumann1994moral}. Therefore, harming others is considered acceptable if the net increase in the well-being of a greater number of people can be established~\cite{conway2013deontological}.

Figure~\ref{fig:ought_sets} illustrates the fundamental differences between the two approaches. Unlike the deontological actor, who chooses a course of action from a predefined set of permitted actions in a given situation, the utilitarian actor chooses the option that maximises utility from among all possible actions, which may go beyond the permitted set according to a moral system.

Psychology claims that this dichotomy is rooted in the way the human mind works.  Dual-process theories, as outlined by~Evans~\cite{evans2003two}, suggest that two distinct cognitive systems underpin human reasoning; an ancient evolutionary relic shared with other animals, and a recent addition to the evolutionary timeline unique to humans.
It suggests that moral responses to dilemmas are informed by both deontology and utilitarianism~\cite{conway2013deontological} when applied to moral decision making. 
Using functional magnetic resonance imaging (fMRI) to study moral dilemmas, Greene and colleagues~\cite{greene2004neural} showed that deontological principles are processed automatically and unconsciously in the brain, whereas utilitarian frameworks are processed explicitly and consciously. There was also evidence of competition between the two systems.

Rationality, even in the context of uncertainty, provides the basis for sound decision-making. One such example is expected utility theory (EUT)~\cite{Mas-Colell1995,gigerenzer2010moral}. However, EUT does not take into account the costs associated with searching.
Simons'~\cite{simon1997models} concept of bounded rationality is a decision-making framework that takes into account the limited resources of the individual in searching for the best course of action, especially since the individual must also consider the benefits and costs. 
Gigerenzer's~\cite{gigerenzer2010moral} emphasis on moral satisficing holds that moral behaviour is the result of the interplay between mind and environment, as it cannot be explained by reference to moral rules or maximisation principles alone. Rather, it is the result of pragmatic social heuristics.

Mattsson and Weibul's~\cite{mattsson2002probabilistic} microeconomic model is an example of bounded rationality outside of game theory. It is a probabilistic model in which the decision maker must explicitly consider disutility and finds the Gibbs distribution as the optimal solution.
Previous models of bounded rationality, such as McFadden~\cite{Mcfadden2012}'s `Quantal Response Equilibrium' model, generate the Boltzmann distribution as the solution as a consequence of Luce's Choice Axiom.~\cite{luce1959individual}.
In line with Hansen and Sargent~\cite{hansen2001robust}, these models are classified as multiplier robust-control problems or, alternatively, if a source probability distribution for world states is provided, they are classified as constraint robust-control problems. 
The latter type of control problem is discussed by Sims~\cite{SIMS2003} in `Rational Inattention Theory', in which the behaviour of economic agents is linked to Shannon's theory of communication~\cite{Berger1971,cover1991elements}.  Denti et al.~\cite{denti2020note} discuss the relationship between rational inattention and rate distortion theory through the systematic use of Markov kernels.
The Information Bottleneck' method, as proposed by Tishby et al.~\cite{tishby2000information}, identifies the greatest possible compression of the input signal while preserving as much information as possible in the output signal.  

Ortega and Braun~\cite{ortega2013thermodynamics} reinterpreted bounded rational decision making by thermodynamic means as a system in which information processing costs are modelled as state changes quantified by free energy differences. They also made a remarkable connection between the coupling constant and drift-diffusion models used to model reaction time experiments.
Building on the thermodynamic formalism for bounded rationality, Genewein et al.~\cite{genewein2015bounded} deduced the formation of abstractions and demonstrated the influence of information processing costs on decision hierarchies.

We build in part on our previous work~\cite{Friedrich2023},  
by using the theory of resource-bounded rationality to formalise ideas from dual-process theory, and consider their implications for both moral and legal systems. Although, formally similar, we interpret resource-bounded rationality as both a cost-based phenomenon and a regularisation-based phenomenon, a perspective also taken by Ortega and Stocker~\cite{ortega2016human}. 
Regularisation, e.g. Tikhonov regularisation, is used to avoid overfitting in empirical risk minimisation, but also in Bayesian statistics~\cite{CaE14,deisenroth2020mathematics}. 

Let us now summarise the content of the paper. 

We begin by discussing Markov kernels in some detail. We talk about group actions and co-sheaves as natural parts of our framework.
We then discuss the basic results from information geometry and how they relate to Markov kernels. This is necessary to clarify the geometry related to resource-bounded rationality in the following sections.
After these general steps, we examine the two main models of bounded rationality in the literature. The first is a robust multiplier control problem, which analyses separate bounded rational decision processes because it corresponds to a single event perspective.
We then study a general constraint robust control problem, which unites separate optimisation problems by using the bottleneck method in conjunction with a source probability on the states of the world. 
We solve this optimisation problem by considering the gradient of mutual information, which gives the utility expansion path.
In the final section, we show how the introduced framework can be used to model the process of restriction of fundamental rights by an authority, as codified in most continental European legal systems. 
We then interpret the deontological and utilitarian components of the formulae describing a bounded rational decision-making process. 
The most important point is that in all cases the coupling constant will remain a free parameter to be determined by the competent legislative or judicial authority.

A more comprehensive analysis of legal theory will be provided in a separate paper.

\begin{figure}
    \centering
    \includegraphics[width=0.6\linewidth]{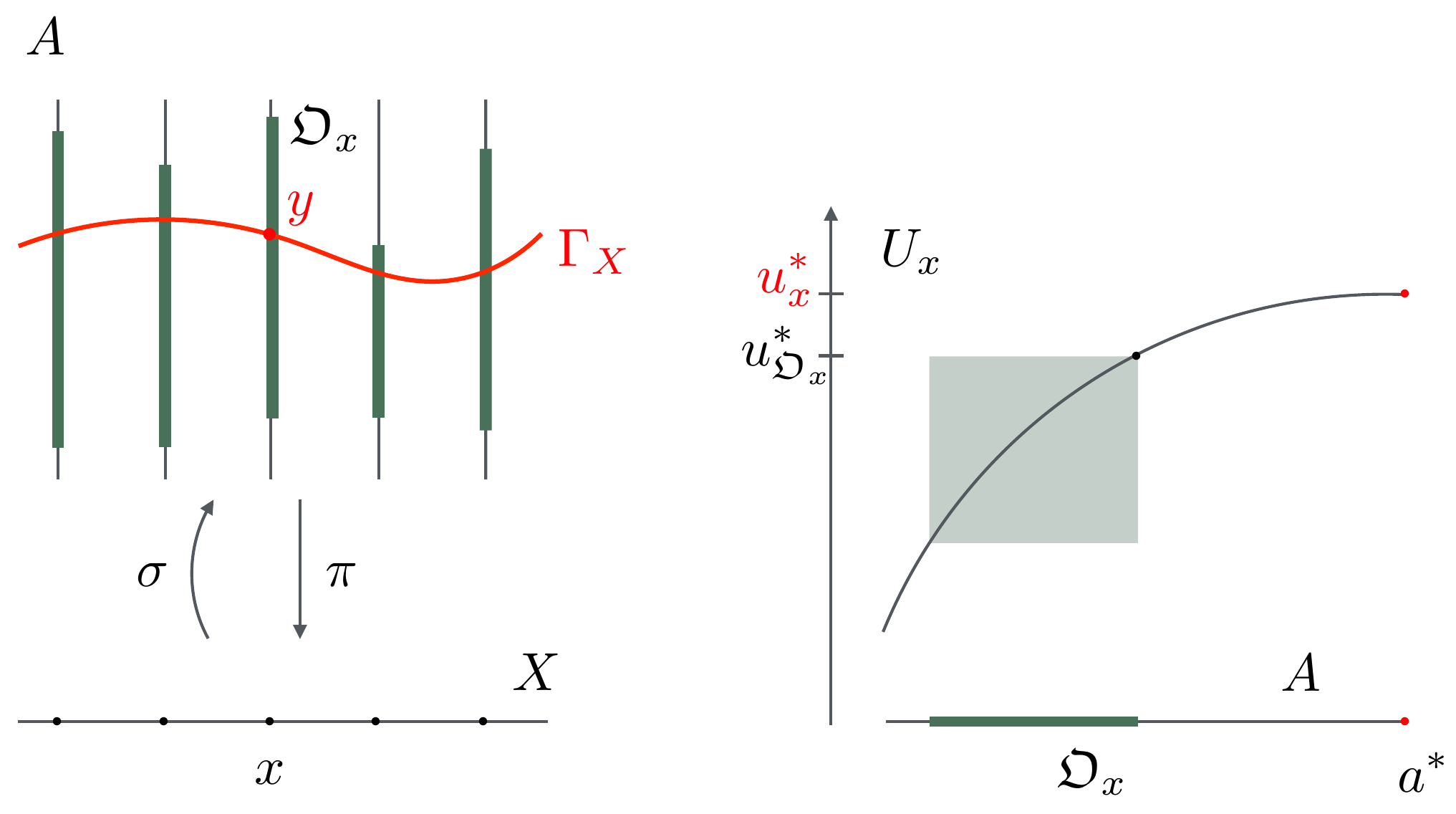}
    \caption{Left panel: Ought set $\mathfrak{O}_x\subset A$ of allowed actions in situation $x$. A particular norm or law corresponds to a co-section $\sigma:X\rightarrow Y$ with $y=\sigma(x)$. $\Gamma_X$ global co-section. Right panel: Utility of all actions. Maximum utility {\color{red} $u_x^*$} achieved at $a^*$, but not eligible because $a^*\notin\mathfrak{O}_x$.}
    \label{fig:ought_sets}
\end{figure}

\section{Varia}
Following Cooter~\cite[p. 217 et seq.]{skurski1983new}, let $P_{ij}$ be the proposition that specifies the action that changes the state $x_i$ to $x_j$, i.e. $P_{ij}.x_i=x_j$.  A rule, $OP_{ij}$, is then a statement of the form that in the state $x_i$ someone should take the action described by $P_{ij}$. So the imperative $OP_{ij}$ encodes the rule maker's preference $x_i\prec x_j$. Let us expand on this.
\subsection{Group actions}
Let $(X,\mathscr{X})$ be a discrete measurable space and $(A,\circ,\mathfrak{e})$ a finite monoid with neutral element $\mathfrak{e}$. 
A (left) monoid action $\alpha$ of $A$ on $X$ is a map $\alpha : A\times X\rightarrow X$, such that for all $s,t\in A$ and $x\in X$, $\alpha(s,\alpha(t,x)) = \alpha(s\circ t,x)$; and for all $x\in X$: $\alpha(\mathfrak{e},x) = x$. We shall write $a.x:=\alpha(a,x)$ for a left action, and refer to $a.x$ as the consequence of action $a$ on state $x$. If $(A,\mathscr{A})$ is a measurable monoid, we require the action to be measurable. If $Y< A$ is a submonoid, then the action restricted to $Y$ is denoted by $\alpha|_Y$. 

Let $u:X\rightarrow\R$ be a $\mathscr{X}$-measurable utility function representing a preference relation $\succsim$ on $X$. Thus, for all $x,y\in X$, $x\succsim y\Leftrightarrow u(x)\geq u(y)$.
Define 
\begin{equation}
\label{eq:utility}
U:X\times A\rightarrow\R,\quad (x,a)\mapsto U(x,a):=u(a.x)-u(x),
\end{equation}
which is the change in utility given the action $a$ on the state $x$. For a fixed value of $x$ we use the notation $U_x$. 

Note that the utility function defines a potential function on the different states of the world. Consequently, different actions leading to the same state will have the same difference in utility. However, one can define a state-dependent cost function for actions that is independent of the utility of the resulting state, cf.~Section~\ref{sec:legal_app}. 

\subsection{Ought co-sheaves}
Set $[N]:=\{0,1,\dots,N\}$ for $N\in\N\cup\{+\infty\}$. Let $X$ be a discrete set and let $(B_i)_{i\in [N]}$ be a (countable) partition of $X$, with blocks $B_i\neq\emptyset$. Define the partition topology $\tau_X:=\tau\{B_i:i\in[N]\}$ and the partition $\sigma$-algebra $\mathscr{X}:=\sigma\{B_i:i\in[N]\}$ generated by the partition. Note that the open sets in $\tau_X$ are also closed, and $\tau_X$ and $\mathscr{X}$ are the same. 

Let $A$ be a (discrete) set. Assign a set $\mathfrak{O}(B_i)\subset A$ to each block $B_i$, and for $V,W\in\tau_X$, $C:=V\cap W$, define $\mathfrak{O}(V\cup W):=\mathfrak{O}(V)\bigsqcup_{C}\mathfrak{O}(W)$ by the pushout. 

Then every $U\in\tau_X$ has an up to order unique representation $U=\bigsqcup_{j=1}^m B_{i_j}$, and $\mathfrak{O}(U)$ has pairs $(i_j,a)$, $a\in\mathfrak{O}(B_{i_j})$, as elements. For $U\subset V$, the inclusion morphism $\iota_{UV}:\mathfrak{O}(U)\rightarrow\mathfrak{O}(V)$ is the inclusion of sets satisfying $\iota_{UU}$, which is the identity map.
This defines a set-valued pre-cosheaf $\mathfrak{O}$ on $\tau_X$ (which is actually a cosheaf)~\cite[V. {\S}1]{Bredon1997}.

Let $\mathfrak{A}$ and $\mathfrak{B}$ be precosheaves on $X$ and $Y$ respectively. A morphism of precosheaves along $f$ consists of a pair $(f,f_{\#})$, where $f:X\rightarrow Y$ is a continuous map and $f_{\#}:f_*\mathfrak{A}\rightarrow\mathfrak{B}$ is a natural transformation~\cite[p. 287]{Bredon1997}. Therefore, for all $U,V\in\tau_Y$, $U\subset V$, the diagram 
\begin{equation*}
\xymatrix{
f_*\mathfrak{A}(V) \ar[r]^{f_{\#V}} & \mathfrak{B}(V) \\
f_*\mathfrak{A}(U)\ar[r]^{f_{\#U}}\ar[u]_{\iota_{UV}^{f_*\mathfrak{A}}} & \mathfrak{B}(U)\ar[u]_{\iota_{UV}^{\mathfrak{B}}} 
}  
\end{equation*}
commutes; note that the direction is opposite to that for sheaves. If $f=\operatorname{id}_X$, it is called a morphism of pre-cosheaves on $X$. If $f:X\rightarrow Y$ is surjective, we call the pair $(f,f_{\#})$ an abstraction, and $\mathfrak{A}$ a specification of $\mathfrak{B}$.

Example: Let $x\mapsto\mathfrak{O}_x:=\mathfrak{O}(\{x\})\subset A$, for $\tau_X:=2^X$ (power set), be the map that associates the actions that should be taken in state $x$. The data is summarised in a policy matrix $\mathfrak{O}=(\mathfrak{o}_{xa})$ of size $|X|\times|A|$, with elements 
$$
\mathfrak{o}_{xa}=\begin{cases}
      & 1,\quad \text{if $a\in\mathfrak{O}_x$}, \\
      & 0,\quad \text{otherwise}.
\end{cases}
$$ 
When rows are selected from overlapping sets in this matrix, the co-sheaf property requires that duplicates are discarded and only one copy is retained. In summary, this construction yields the ought co-sheaf $\mathfrak{O}_X$ of actions.

Example: Let $K:(X,\mathscr{X})\rightarrow(Y, \mathscr{Y})$ be a Markov kernel (stochastic matrix), cf.~Section~\ref{sec:Markov}. By setting $\mathfrak{O}_x :=\operatorname{supp}(K_x)$ and extending it by the co-product, we obtain a co-sheaf on $X$.
Conversely, given a policy matrix $\mathfrak{O}$, it determines the support of the stochastic kernels that can be considered in a probabilistic control problem, i.e. $\operatorname{supp}(K_x)\subset\mathfrak{O}_x$ must hold. This is discussed in more detail in Section~\ref{sec:robust}.

\section{Markov kernels}
\label{sec:Markov}

Let $(X,\mathscr{X})$ and $(Y,\mathscr{Y})$ be measurable spaces. A Markov kernel from $(X,\mathscr{X})$ to $(Y,\mathscr{Y})$, is a map $K:X\times\mathscr{Y} \rightarrow[0,1]$, such that
\begin{eqnarray}
\label{eq:markov_kernel_1}
    &&K(x,\cdot)~\text{is a probability measure on $\mathscr{Y}$ for all $x\in X$};\\\label{eq:markov_kernel_2}
&&K(\cdot, B)=: K_B:X\rightarrow[0,1]~\text{is $\mathscr{X}$-measurable for all $B\in\mathscr{Y}$}.
\end{eqnarray}
Note that if $X$ is discrete and $\mathscr{X}:=2^X$, then~(\ref{eq:markov_kernel_2}) is always satisfied, since 
$\{K_B(x)\leq r\}\in 2^X$ for all $r\in [0,1]$ and $B\in\mathscr{Y}$.

For a fixed probability measure $\nu\in\Delta(Y)$ (from the space of probability measures on $(Y,\mathscr{Y})$, cf.~Section~\ref{sec:Info-geo}), the constant kernel $K_{\nu}$ is defined as $K_{\nu}(x,B):=\nu(B)$ for all $x\in X$ and $B\in\mathscr{Y}$.
Set $K_x:=K(x,\cdot)$ and $k(x,y):=K(x,\{y\})$ if $\{y\}\in\mathscr{Y}$. 
In the finite/discrete case, it is represented by a stochastic matrix $(K_{ij})$, where
\begin{equation}
K_{ij}\geq 0,\,\,\text{for all $i\in X, j\in Y$, and}\quad\sum_{j\in Y} K_{ij}=1~\text{for all $i\in X$}.
\end{equation}
A Markov kernel determines a conditional probability via $p(y|x):=K(x, \{y\})$. 
We denote (if the context is clear) by $\mathfrak{K}(X,Y)$ the set of all stochastic kernels from $(X,\mathscr{X})$ to $(Y,\mathscr{Y})$.

The semi-direct product $\rtimes:\Delta(X)\times\mathfrak{K}(X,Y)\rightarrow\Delta(X\times Y)$ between probability measures and Markov kernels, yields a probability measure $\pi:=\P\rtimes K$ on the product space $(X\times Y,\mathscr{X}\otimes\mathscr{Y})$. On a base, it is defined by 
\begin{equation}
\label{eq:semidirect_prod}
(\P\rtimes K)(A\times B):=\int_A K(x,B)d\P(x),\quad A\in\mathscr{X}, B\in\mathscr{Y}.
\end{equation}
For a constant kernel $K_{\nu}$ we get the product measure $\P\otimes\nu$, because 
$$
(\P\rtimes K_{\nu})(A\times B)=\int_A K_{\nu}(x,B)d\P(x)=\int_A\nu(B)d\P(x)=\nu(B)\cdot\P(A)=(\P\otimes\nu)(A\times B).
$$
For a finite probability measure $\P=(p_1,\dots,p_m)$ and stochastic matrix $K=(K_{ij})$, we have 
$$
(\P\rtimes K)_{ij}=p_i\cdot K_{ij}.
$$
\begin{figure}
    \centering
    \includegraphics[width=0.7\linewidth]{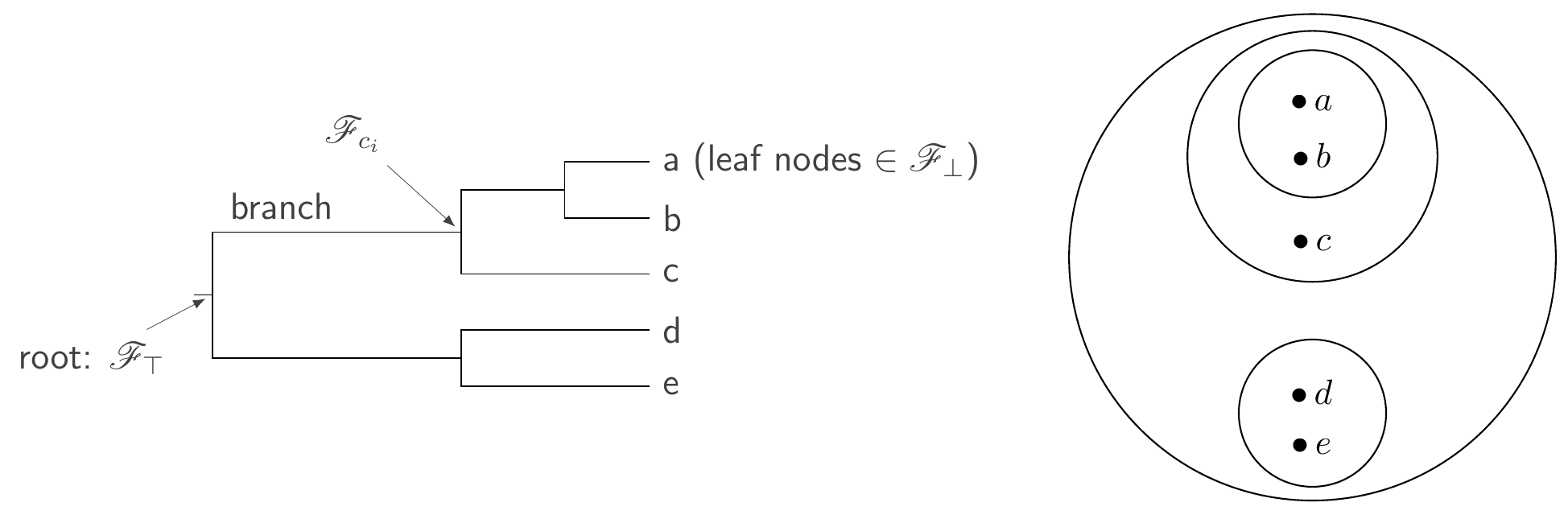}
    \caption{Hierarchical clustering and filtration with corresponding bracketing $\{\{\{\{A\},\{B\}\},\{C\}\},\{\{D\},\{E\}\}\}$. The root node corresponds to the trivial algebra $\mathcal{F}_{\top}$, and the leaf nodes correspond to the power set algebra $\mathcal{F}_{\bot}$.}
    \label{fig:clustering}
\end{figure}
\noindent 
A measurable space $(X, \mathscr{X})$ is said to be standard or Borel if there exists a Borel set $B \in \mathcal{B}(\R)$ such that $(X, \mathscr{X})$ is isomorphic to $(B, \mathcal{B}|_B)$. Every standard measurable space is isomorphic to either the Borel $\sigma$-algebra on the interval $[0,1]$, all finite sets of the form $[1,\dots,n]$, or countable sets, with the discrete $\sigma$-algebra~\cite[p. 11]{ccinlar2011probability}.

Let $\operatorname{pr}_X:X\times Y\rightarrow X$, $\operatorname{pr}_X(x,y):=x$ be the projection onto the $X$-component. The $X$-marginal of a probability measure is given by 
$({\operatorname{pr}_{X}})_*:\Delta(X\times Y)\rightarrow\Delta(X)$, $\pi\mapsto(\operatorname{pr}_X)_*\pi=:\pi_X$. The definition and the properties for the $Y$-projection $\operatorname{pr}_{Y}$ are analogous.  
For $\pi\in\Delta(X\times {\color{blue}Y})$, where ${\color{blue}(Y,\mathscr{Y})}$ is Borel, the `Disintegration Theorem' (\cite[Theorem 2.18]{ccinlar2011probability}) states that there exist $\mu\in\Delta(X)$ and $K\in\mathfrak{K}(X,Y)$ such that $\mu\rtimes K=\pi$. Then the solution is written (symbolically) as 
$$
K=\frac{\pi}{\pi_X},
$$ 
with $\mu=\pi_X$.
In terms of conditional probabilities we have (in components)
$$
\pi(y|x)=\frac{\pi(x,y)}{\pi(x)}.
$$
The push-forward by a Markov kernel $K_*:\Delta(X)\rightarrow\Delta(Y)$, is defined as
\begin{equation}
\label{eq:Markov_map}
   K_*\P(B):=\int_X K(x,B)d\P(x),\quad B\in\mathscr{Y}.
\end{equation}
It gives the $Y$-marginal of $\pi:=\P\rtimes K$, i.e. $K_*\P=\pi_Y$, as shown in the commutative diagram
\begin{equation}
\label{eq:Y_marginal}
\xymatrix{
\Delta(X) \ar[r]^{\!\!\!\!\!\rtimes K} \ar[dr]_{K_*} & \,\,\,\Delta(X\times Y) \ar[d]^{{\operatorname{pr}_{Y}}_*} \\
 & \Delta(Y)
}  
\end{equation}
\begin{lem} 
\label{rem:kernel}
The semi-direct product $\rtimes$ is convex-bilinear. Furthermore, $\mathfrak{K}(X,Y)$ is a convex subset of the Euclidean space $\R^{X\times Y}$. For finite $X,Y$, it is compact.
\end{lem}
\begin{proof}
For $K,K_1,K_2\in\mathfrak{K}(X,Y)$, $\mu,\nu\in\Delta(X)$ and $t\in[0,1]$, we have $(1-t)K_1+tK_2\in\mathfrak{K}(X,Y)$ and $(1-t)\mu+t\nu\in\Delta(X)$.
Thus
\begin{equation*}
\begin{split}
    ((1-t)(\P\rtimes K_1)+t(\P\rtimes K_2))(A\times B)=(1-t)\int_A K_1(x,B)d\P(x)+t\int_A K_2(y,B)d\P(x)\\
    =\int_A((1-t)K_1(x,B)+tK_2(x,B))d\P(x)=\left(\P\rtimes\left((1-t)K_1+tK_2\right)\right)(A\times B),
\end{split}
\end{equation*}
and
\begin{equation*}
\begin{split}
    ((1-t)(\mu\rtimes K)+t(\nu\rtimes K))(A\times B)=\int_A K(x,B)(1-t)d\mu(x)+\int_A K(y,B)td\nu(x)\\
    =\int_AK(x,B)((1-t)d\mu(x)+td\nu(x))=\left(((1-t)\mu+t\nu)\rtimes K\right)(A\times B).
\end{split}
\end{equation*}
In the finite-dimensional case, $\mathfrak{K}(X,Y)$ is isomorphic to $\prod_{x\in X}\Delta(Y)$, and the product of compact spaces is compact according to the `Tube Lemma'.\footnote{\url{https://en.wikipedia.org/wiki/Tube_lemma}}
\end{proof}

\subsection{Rate distortion}
A triple $(X,K,Y)$, consisting of a finite input alphabet $X$ (source), a finite output alphabet $Y$ (target), and a fixed Markov kernel $K:X\rightarrow Y$, is called a discrete memoryless channel (DMC)~\cite{topsoe1972new,Berger1971,cover1991elements}. 

For $p,q\in\Delta(X)$, $p$ is absolutely continuous w.r.t. $q$, denoted as $p\ll q$, if and only if $\operatorname{supp}(p)\subset\operatorname{supp}(q)$, where $\operatorname{supp}$ denotes the support of the measure.

The Kullback-Leibler divergence $D_{KL}(p||q)$ of two finite-dimensional distributions $p,q\in\Delta_{n}$, is given by
\begin{equation}
D_{KL}(p||q):=
\begin{cases}
      & \sum_{i=0}^n p_i\ln\frac{p_i}{q_i},\quad \text{if $p\ll q$}, \\
      & +\infty,\quad\text{else},
\end{cases}
\end{equation}
with the conventions $0\ln\frac{0}{q}=0$, $p\ln\frac{p}{0}=+\infty$ and $0\ln\frac{0}{0}=0$ in place. It satisfies $D_{KL}(p,q)\geq 0$ and $D_{KL}(p||q)=0$ if and only if  $p=q$,~\cite[p. 18]{amari2016information}.
For the Dirac measure at $i$ and $q\in\Delta^{°}_n$, i.e. a measure with full support (cf.~(\ref{eq:interior_simplex})), we have $D_{KL}(\delta_i||q)=-\ln q_i$, while $D_{KL}(q||\delta_i)=+\infty$, illustrating the asymmetry of the relative entropy. 

The mutual information $I(\P;K)$ between $\P\in\Delta(X)$ and $K\in\mathfrak{K}(X,Y)$ is defined as
\begin{equation}
\label{eq:mutual_info}
    I(\P;K):=D_{\operatorname{KL}}(\P\rtimes K||\P\otimes K_*\P),
\end{equation}
or, equivalently, $I(\pi):=D_{\operatorname{KL}}(\pi||\pi_X\otimes\pi_Y)$, for $\pi:=\P\rtimes K$.
Given random variables  $X\sim p$ and $Y\sim q$, the mutual information\footnote{Let $\pi_{(X,Y)}$ be the joint distribution of $X$ and $Y$. Then $I(X,Y):=D_{\operatorname{KL}}(\pi_{(X,Y)}||p\otimes q)$.} 
satisfies $I(X,Y)\geq0$ and $I(X,Y)=0$ if and only if $X$ and $Y$ are independent, cf.~\cite[p. 27]{cover1991elements} or~\cite[p. 46]{amari2016information}.

The channel capacity $C_K$ of a Markov kernel $K$ is defined as:
\begin{equation}
    C_K:=\max_{\P\in\Delta(X)} I(\P;K).
\end{equation}
\subsection{Renormalisation}
Clustering of data is modelled by sub-$\sigma$-algebras, and more generally by filtrations, cf. Figure~\ref{fig:clustering}. If the initial $\sigma$-algebra is replaced by a coarser one by means of coarse graining, then only events on a larger scale can be observed. This process corresponds to the `emergence of abstractions'~\cite[p. 5]{genewein2015bounded}.  It is equally important to recognise that any form of statistical manipulation can, at best, preserve the same amount of information.

Given coarse graining (surjective) maps $f:[m]\rightarrow[\tilde{m}]$ and $g:[n]\rightarrow[\tilde{n}]$ with $\tilde{m}\leq m$ and $\tilde{n}\leq n$ respectively, then for a stochastic matrix $K=(k_{ij})_{\substack{i\in[m] \\ j\in[n]}}$, the renormalised matrix $\tilde{K}=(\tilde{k}_{\mu\nu})_{\substack{\mu\in[\tilde{m}] \\ \nu\in[\tilde{n}]}}$ is given by
\begin{equation}
\tilde{k}_{\mu\nu}=\frac{1}{\P(f^{-1}(\mu))}\sum_{i\in f^{-1}(\mu)}\sum_{j\in g^{-1}(\nu)}k_{ij}.
\end{equation}
The above shows that in the discrete case, we can change the order of summation without changing the result. Note that the conditioning smears the support of each probability measure in the kernel. We will now discuss the general steps involved in this procedure.

Let $g:Y\rightarrow Z$ be a (surjective) $\mathscr{Y}/\mathscr{Z}$-measurable map and $K\in\mathfrak{K}(\mathscr{X},\mathscr{Y})$. 
The push-forward of $K$ by $g$, denoted $g_*K$, is defined for all $x\in X$ and $A\in\mathscr{Z}$ as
\begin{equation}
(g_*K)(x,A):=K(x,g^{-1}(A)).
\end{equation}
In particular, for every $A\in\mathscr{Z}$, $(g_*K)_{A}(x):=K(x,g^{-1}(A))$ is $\mathscr{X}/\mathcal{B}([0,1])$-measurable, and hence $g_*K\in\mathfrak{K}(\mathscr{X},\mathscr{Z})$.

We consider conditional Markov kernels $\E[K|\mathscr{F}]$ via regular conditional distributions.
Let $(X,\mathscr{X},\P)$ be a probability space, $(Y,\mathscr{Y})$ a Borel space, and $K\in\mathfrak{K}(\mathscr{X},\mathscr{Y})$ a Markov kernel. 
Consider $(X\times Y,\mathscr{X}\otimes\mathscr{Y},\pi)$, the joint probability space, where $\pi:=\P\rtimes K$, and let $\operatorname{pr}_X$ and $\operatorname{pr}_Y$ be the projections onto the $X$ and $Y$ components, respectively.
For $(Z,\mathscr{Z})$, an arbitrary measurable space, let $f:X\rightarrow Z$ be $\mathscr{X}/\mathscr{Z}$-measurable, and define $\xi:X\times Y\rightarrow Z$ as $\xi:=f\circ\operatorname{pr}_X$, i.e. $(x,y)\mapsto f(x)$. Then there exists a regular conditional distribution $f_*^{\P}K\in\mathfrak{K}(\mathscr{Z},\mathscr{Y})$ of $\operatorname{pr}_Y$ given $\xi$~(cf.~\cite[2.18 Theorem]{ccinlar2011probability}), such that (cf.~Figure~\ref{img:reg_conditional_prob})
\begin{equation}
f_*^{\P}K(z,B)=\begin{cases}
      & \pi(\operatorname{pr}_Y^{-1}(B)|\xi=z)\quad\text{if $\pi(\xi^{-1}(z))>0$ }, \\
      & \pi(\operatorname{pr}_Y^{-1}(B))\quad\text{otherwise},
\end{cases}
\end{equation}
holds, for $f_*\P$-almost all $z\in Z$ and for all $B\in\mathscr{Y}$; i.e.
$$
f_*^{\P}K(z,B)\equiv \E[\mathbbm{1}_{\operatorname{pr}_Y^{-1}(B)}|\xi=z]=\frac{1}{f_*\P(z)}\int\limits_{f^{-1}(z)\times Y}\!\!\!\!\!\!\mathbbm{1}_B(y)\,\mathrm{d}\pi(x,y)\pmod{f_*\P}.
$$
The regular conditional probability is therefore given by the disintegration/factorisation of the product measure, i.e. we write symbolically
$$
f_*^{\P}K=\frac{\P\rtimes K}{f_*\P}.
$$
So we have the following commutative diagrams for kernels and measures:
\begin{equation}
\xymatrix{
(X,\mathscr{X})\ar[r]^{K}\ar[dr]_{g_{*}K} & (Y,\mathscr{Y})\ar[d]^{g}&&&(X,\mathscr{X})\ar[r]^K\ar[d]_f&(Y,\mathscr{Y})\,\,\,\text{(Borel)}\\  
& (Z,\mathscr{Z})&&&(Z,\mathscr{Z})\ar[ur]_{f_*^{\P}K}&
}  
\end{equation}

\begin{minipage}{0.4\textwidth}
These maps correspond to the following covariant functors
\begin{eqnarray*}
{\color{red}g_*}&:&\mathfrak{K}({\color{blue}\mathscr{X}},\mathscr{Y})\rightarrow\mathfrak{K}({\color{blue}\mathscr{X}},\mathscr{Z}),\\
 &&K(x,-)\mapsto K(x,{\color{red}g^{-1}}(-))\\
{\color{red}f^{\P}_*}&:&\mathfrak{K}(\mathscr{X},{\color{blue}\mathscr{Y}})\rightarrow\mathfrak{K}(\mathscr{Z},{\color{blue}\mathscr{Y}}),\\
 &&K(-,B)\mapsto (z\mapsto\E_{\color{red}\P}[K_{B}|{\color{red}f}=z])
\end{eqnarray*}
$x\in X$, $z\in Z$, $B\in\mathscr{Y}$, where {\color{red}$\P$} enters as a parameter. 
\end{minipage}%
\begin{minipage}{0.6\textwidth}
\begin{center}
    \includegraphics[width=0.8\textwidth]{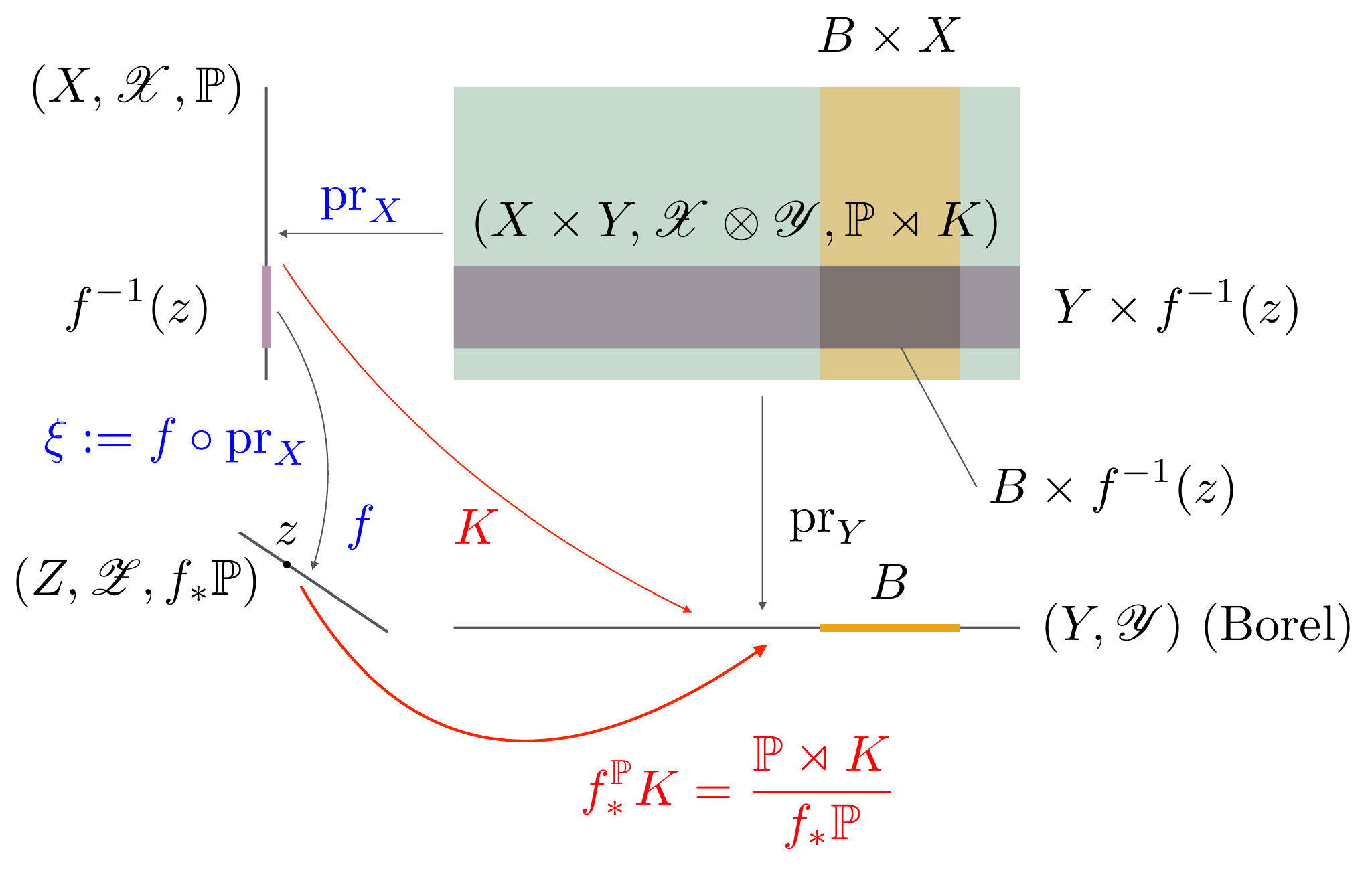}
    \captionof{figure}{Conditional Markov kernel}
    \label{img:reg_conditional_prob}
\end{center}
\end{minipage}

Let us briefly consider the impact of the operations introduced on information. 
The sequence of random variables 
$X\sim p$, $Y\sim K_*p$ and $Z=g\circ Y\sim (g_*K)_*p$ form a Markov chain $X\rightarrow Y\rightarrow g\circ Y$. The `Data Processing Inequality' implies that 
\begin{equation}
\label{eq:data_processing_inequality}
I(\P,K_*\P)\geq I(\P,(g_*K)_*\P),
\end{equation}
or, in a more conventional form, $I(X,Y)\geq I(X,g(Y))$. Therefore, it is impossible to increase the amount of information contained in $X$ by any function of the data~\cite[Corollary to 2.8.1]{cover1991elements}.

Consider the joint probability space $(X\times Y,\mathscr{X}\otimes\mathscr{Y},\P\rtimes K)$. 
Let $(F,\mathscr{F})$ and $(G,\mathscr{G})$ be discrete measurable spaces, $(F\times G,\mathscr{F}\otimes\mathscr{G})$ their product, and 
$f:X\rightarrow F$ and $g:Y\rightarrow G$ are surjective measurable maps.
Then $f\times g:X\times Y\rightarrow F\times G$ is $\mathscr{X}\otimes\mathscr{Y}/\mathscr{F}\otimes\mathscr{G}$-measurable, since $(f\times g)^*(\mathscr{F}\otimes\mathscr{G})=f^*\mathscr{F}\otimes g^*\mathscr{G}\subset\mathscr{X}\otimes\mathscr{Y}$. 
 \begin{lem}
With the above conditions, the push-forward of $\P\rtimes K$ by $f\times g$ to $\mathscr{F}\otimes\mathscr{G}$, satisfies
$$
(f\times g)_*(\P\rtimes K)=f_*\P\rtimes g_*(f_*^{\P}K).
$$
In addition, $g_*$ and $f_*^{\P}$ commute, i.e. $g_*f^{\P}_*K=f_*^{\P}g_*K$, as shown in the diagram:
$$
\xymatrix{
(X,\mathscr{X},\P)\ar[d]_f\ar[r]^{K} & (Y,\mathscr{Y})\ar[d]^{g}\\  
(F,\mathscr{F})\ar[r]& (G,\mathscr{G})
}  
$$
\end{lem}
\begin{proof}
Let $\pi:=(f\times g)_*(\P\rtimes K)$. By~\cite[2.18~Theorem]{ccinlar2011probability} there exists a $\pi_X$-unique factorisation $\pi=\pi_X\rtimes K_{\pi}$, where $K_{\pi}:=\frac{\pi}{\pi_X}\in\mathfrak{K}(\mathscr{F},\mathscr{G})$. If we unwind the definition of the different quantities, we find that $\pi_X=f_*\P$ and then by calculation and using~(\ref{eq:semidirect_prod}) the rest of the statement is shown.
\end{proof}
Note that for $X$ and $Y$ discrete, the map induced by $f$ and $g$, $f_*^{\P}\circ g_*:\mathfrak{K}(X,Y)\rightarrow\mathfrak{K}(F,G)$ is surjective. This is a consequence of the point-wise application of the following:
\begin{lem}
\label{lem:projection_measure}
Let $(X,2^X)$ and $(Y,2^Y)$ be discrete measurable spaces, equipped with the power set $\sigma$-algebra and $f:X\rightarrow Y$ surjective. Then $f_*:\Delta(X)\rightarrow\Delta(Y)$ is surjective.
\end{lem}
\begin{proof}
Let $\iota:Y\hookrightarrow X$ be an embedding such that $\iota(y)\in f^{-1}(y)$ for all $y\in Y$ and $q\in\Delta(Y)$. Then $\iota_*q\in\Delta(X)$ and $f_*(\iota_*q)=q$. In fact, for $y\in Y$ we have $f_*(\iota_*q)(y)=\iota_*q(f^{-1}(y))=q(\iota^{-1}(f^{-1}(y)))=q(y)$.
\end{proof}
For fixed $\mu\in\Delta(X)$ and $\nu\in\Delta(Y)$, define $\mathfrak{K}_{\mu\nu}: =\{K\in\mathfrak{K}(\mathscr{X},\mathscr{Y})~|~K_*\mu=\nu\}$, and similarly for the set $\mathfrak{K}_{\nu\mu}$. 

Given $K\in\mathfrak{K}_{\mu\nu}$, we define its reciprocal $K^{-1}\in\mathfrak{}K_{\nu\mu}$ by disintegration, i.e. 
$$
K^{-1}:=\frac{\mu\rtimes K}{K_*\mu},
$$
which is $\nu$-a.s. unique.
Then the set of couplings $\Gamma(\mu,\nu)$, as a consequence of disintegration, is in bijective correspondence with $\mathfrak{K}_{\mu\nu}$, given by $\pi\mapsto K_{\pi}$, which is $\mu$-a.s. well defined, cf.~\cite[2.18~Theorem]{ccinlar2011probability}. Furthermore, for all $K\in\mathfrak{K}_{\mu\nu}$ and $\kappa\in\mathfrak{K}_{\nu\mu}$ we have $\kappa_*(K_*\mu)=\mu$ and $K_*(\kappa_*\nu)=\nu$. Note also the relationship with congruent Markov kernels~\cite[Theorem 5.2]{ay2017information}.

Bayes' factorisation for conditional probabilities states in components that
$$
\pi(y|x)\pi(x)=\pi(x,y)=\pi(x|y)\pi(y).
$$

\begin{prop}[Bayes' disintegration/factorisation]
Let $(X,\mathscr{X},\mu)$ and $(Y,\mathscr{Y},\nu)$ be Borel. Then for all $\pi\in\Gamma(\mu,\nu)$ there exists an a.s.-unique pair $(K_{\pi}, K_{\pi}^{-1})$ of reciprocal Markov kernels such that  
$$
\mu\rtimes K_{\pi}=\pi=\nu\rtimes K_{\pi}^{-1}.
$$
\end{prop}
\begin{proof}
By assumption, both measurable spaces are Borel. Then by~\cite[2.18~Theorem]{ccinlar2011probability} and by changing the role of $X$ and $Y$ respectively, the claim follows.
\end{proof}

Renormalisation reduces not only dimensions but also divergence, as follows from the principle of `Information Monotonicity',
$$
D_{\operatorname{KL}}(\P\rtimes K||\P\otimes K)\geq D_{\operatorname{KL}}(f_*\P\rtimes g_*(f_*^{\P}K)||f_*\P\otimes g_*(f_*^{\P}K)),
$$
with equality, if the statistic is sufficient~\cite[p. 53 et seq.]{amari2016information}.

Let $U:X\times Y\rightarrow\R$ be a $(\mathscr{X}\otimes\mathscr{Y})$-measurable (utility) function. The conditional expectation w.r.t. $\mathscr{F}\otimes \mathscr{G}$ is then given by $\E[U|f^{*}\mathscr{F}\otimes g^*\mathscr{G}]$. In particular, the coarse-grained version of~(\ref{eq:constrained_information}) is
\begin{equation}
\label{eq:coarse_grained_eq}
\argmax_{\tilde{K}\in\mathfrak{K}(\tilde{X},\tilde{Y})}\E_{f_*\P\rtimes \tilde{K}}[\E[U|\mathscr{F}\otimes\mathscr{G}]]-\frac{1}{\beta} D_{\operatorname{KL}}(f_*\P\rtimes \tilde{K}||f_*\P\otimes \tilde{K}).
\end{equation}
The expression~(\ref{eq:coarse_grained_eq}) is well defined, because by~Lemma~\ref{lem:projection_measure} we have surjectivity.

\section{Some information geometry}
\label{sec:Info-geo}
Let us recall some facts from information geometry that we will use. Standard references are Amari~\cite{amari2016information} and Ay et al.~\cite{ay2017information}. Pistone~\cite{pistone2020information} provides an introduction to the information geometry of the probability simplex.

For a set $X$ with $n+1$ elements, the set of all probability functions $p:X\rightarrow[0,1]$, denoted as $\Delta(X)$, forms an $n$-dimensional simplex $\Delta_n$, i.e.  
\begin{equation}
    \Delta(X):=\left\{p=(p_0,p_1,\dots,p_n)~|~\sum_{i=0}^n p_i=1,\, p_i\geq0, \forall i\right\}.
\end{equation}
The $p_i$ are the barycentric coordinates of $p$ with respect to the canonical base $\{\delta_i\}_{i=1}^n$, where $\delta_i$ is the Dirac supported measure at vertex $i$.

The interior of $\Delta(X)$, denoted by $\Delta^{\circ}(X)$, is given by 
\begin{equation}
\label{eq:interior_simplex}
    \Delta^{\circ}(X):=\left\{p=(p_0,p_1,\dots,p_n)~|~\sum_{i=0}^n p_i=1,\, p_i>0, \forall i\right\}.
\end{equation}

Let $\mathscr{S}_0(X):=\{v:X\rightarrow\R~|~\sum_{x\in X}v(x)=0\}$ be the $\R$-algebra of the functions on $X$ with mean zero.
The tangent space $T_p\Delta^{\circ}(X)$ at $p$ is given by
\begin{equation}
   T_p\Delta^{\circ}(X)=\{p\}\times\mathscr{S}_0(X),
\end{equation} 
and the collection of all tangent spaces forms the tangent bundle $T\Delta^{\circ}(X)$.

The Bregman-Voronoi ball~\cite{boissonnat2010bregman} $B_r(q)$, centred at $q$ and of radius $r$, is defined as
    \begin{equation}
    \label{eq:Bregman_ball}
        B_r(q):=\{p\in\Delta(\Omega)~|~D_{\operatorname{KL}}(p||q)\leq r\}.
    \end{equation}
It is a convex set~\cite[Thm. 2.7.2]{cover1991elements}. The Bregman sphere of radius $r$ is given by
    \begin{equation}
    \label{eq:Bregman_sphere}
        \partial B_r(q):=\{p\in\Delta(Y)~|~D_{\operatorname{KL}}(p||q)=r\}.
    \end{equation}

Given two points $p,q\in\Delta^{\circ}(X)$,\footnote{It should be noted that for the $e$-geodesic to be well-defined, the points in question must be interior points.} there exist two types of geodesics curves connecting them~\cite[p. 81]{ay2017information}.

The $m$-geodesic $\gamma^{(m)}_{p,q}:[0,1]\rightarrow\Delta^{\circ}(X)$ is defined as
$$
\gamma^{(m)}_{p,a}(t)=(1-t)p+tq,\quad t\in[0,1].
$$
The $e$-geodesic $\gamma^{(e)}_{p,q}:[0,1]\rightarrow\Delta^{\circ}(X)$ is defined as
$$
\ln\gamma^{(e)}_{p,q}(t)=(1-t)\ln p+t\ln q-\psi(t),\quad t\in[0,1],
$$
with $\ln p$ defined point-wise for each $x\in X$. The $t$-dependent normalisation factor $\psi(t)$ transforms $\gamma^{(e)}_{p,q}(t)$ into a probability mass function, with the factor given by $$
\psi(t):=\ln\int_{X}p(x)^{1-t}q(x)^tdx.
$$
A subset $S\subset\Delta(X)$ is $m$-convex if, for all $p,q\in S$ and $t\in[0,1]$, $\gamma_{p,q}^{(m)}(t)\in S$. It is $e$-convex if, for all $p,q\in S$ and $t\in[0,1]$, $\gamma^{(e)}_{p,q}(t)\in S$, cf.~\cite[p. 46]{hino2024geometry}. 
\begin{lem}
Every $m$- or $e$-convex set $S\subset\Delta(X)$ is simply connected.
\end{lem}
\begin{proof} 
Let $\gamma:\mathbb{S}^1\rightarrow S$ be a continuous closed curve in $S$ with $p:=\gamma(0)$.
Consider the logarithmic homotopy
$\mathbb{S}^1\times[0,1]\rightarrow S$, $(s,t)\mapsto (1-t)\ln p+t\ln\gamma(s)-\psi(s,t)$
where $\psi(s,t):=\ln\int_Xp(x)^{1-t}\gamma(s)^tdx$. 
The $m$-connected case is shown analogously.
\end{proof}
A curve $\gamma(t)$ is orthogonal to the set $S$, if there exists a $\tau_0\in\R$ such that the tangent vector $\dot{\gamma}(\tau_0)$ is orthogonal to $T_{\gamma(\tau_0)}S$ (w.r.t. to the Fisher metric).
\subsection{Exponential families}
An exponential family is a statistical manifold of the form
\begin{equation*}
\mathscr{M}_{\exp}:=\left\{p(x;\theta)=c(x)\cdot\exp\left(\sum_{i=1}^n\theta_i t_i(x)-\psi(\theta)\right)~|~\theta=(\theta_1,\dots,\theta_n)\in\boldsymbol{\theta}\right\}.
\end{equation*}
It consists, with $x$ a random variable, of linearly independent functions $t_i(x)$ (observables), a positive function $c(x)$, $\boldsymbol{\theta}\subset\R^n$ the parameter set and $\psi(\theta)$ the normalisation factor (partition function) necessary to obtain a probability measure.

Note that the partition function $\psi(\theta)$ is convex in $\theta$, cf.~\cite{amari2016information,ay2017information}. 
Furthermore, every exponential manifold is $e$-convex, cf. e.g.~\cite{hino2024geometry}; although they can also be $m$-convex~\cite[Theorem 2.4]{ay2017information}. In Section~\ref{sec:Infogeo_Markov_kernels} we will encounter a simultaneously $e$- and $m$-convex set.

The Boltzmann-Gibbs distribution~(\ref{eq:Boltzmann_Gibbs}) is an example of a one-parameter exponential family, cf.~\cite{amari2016information,ay2017information,naudts2009q}. 

\subsection{Information geometry connected to Markov kernels}
\label{sec:Infogeo_Markov_kernels}
In this section we focus exclusively on finite-dimensional spaces.

For $\P\in\Delta(X)$, $\nu\in\Delta(Y)$, $K_C\in\mathfrak{K}(X,Y)$ and $r\geq0$, define the following subsets of $\Delta(X\times Y)$: 
\begin{eqnarray*}
    \mathscr{M}_{\P}&:=&\{\P\rtimes K~|~K\in\mathfrak{K}(X,Y)\},\\
    \mathscr{M}_{K_C}&:=&\{\mu\rtimes K_C~|~\mu\in\Delta(X)\},\\
    \mathscr{E}_{\perp\!\!\!\perp}&:=&\{\mu\otimes\nu~|~\mu\in\Delta(X), \nu\in\Delta(Y)\},\\\mathscr{E}_r&:=&\{\pi\in\Delta(X\times Y)~|~I(\pi_X,\pi_Y)=r\},\\
    \Gamma(\P,\nu)&:=&\{\pi\in\Delta(X\times Y)~|~\pi_X=\P, \pi_Y=\nu\}.
\end{eqnarray*}
Recall that $\pi\in\Delta(X\times Y)$ is called a coupling of $\mu\in\Delta(X)$ and $\nu\in\Delta(Y)$ if $({\operatorname{pr}_X})_*\pi=\mu$ and $({\operatorname{pr}_Y})_*\pi=\nu$.
\begin{lem}
The subset of all product measures $\mathscr{E}_{\perp\!\!\!\perp}$ is $e$-convex, i.e. an exponential family. For $r>0$, the hypersurface of constant mutual information $\mathscr{E}_r$ consists of two connected components. The intersection $\mathscr{E}_0(\P):=\mathscr{E}_{\perp\!\!\!\perp}\cap\mathscr{M}_{\P}$ is both $e$- and $m$-convex. 
\end{lem}
Note that $\mathscr{M}_{\P}$, with the exception of $\mathscr{E}_0(\P)$, is not $e$-convex, as can be shown by considering examples, cf.~Figure~\ref{fig:e_geodesics}.

\begin{proof}
Consider two product measures $\rho_1:=p_1\otimes q_1$ and $\rho_2:=p_2\otimes q_2$ in $\mathscr{E}_{\perp\!\!\!\perp}$. In this case, the partition function can be expressed as the sum of two separate partition functions
$$
\psi(t)=\psi_X(t)+\psi_Y(t)=\ln\int_{X\times Y}p_1(x)^{1-t}p_2(x)^t q_1(y)^{1-t}q_2(y)^tdxdy.
$$
The first claim now follows from
$$
\ln \rho_t(x,y)=(1-t)\ln p_1(x)+t\ln p_2(x)-\psi_X(t)+(1-t)\ln q_1(y)+t\ln q_2(y)-\psi_Y(t).
$$
The statement about the connectivity follows from the geometric characterisation, as further discussed in~(\ref{eq:mini_mutual_info}).

For the last statement, recall that for a constant kernel $K_{\nu}$, $\nu\in\Delta(Y)$, we have $\P\rtimes K_{\nu}=\P\otimes\nu$. Consider, $\ln \rho_t(x,y):=(1-t)\ln p(x)\mu(y)+t\ln p(x)\nu(y)-\psi_Y(t)$, with $\psi_Y(t):=\ln\int_Y \mu(y)^{1-t}\nu(y)^tdy$. Then, 
\begin{eqnarray*}
    \ln \rho_t(x,y)&=&(1-t)\ln p(x)+t\ln p(x)+(1-t)\ln\mu(y)+t\ln\nu(y)-\psi_Y(t)\\
    &=&\ln p(x)+(1-t)\ln \mu(y)+t\ln\nu(y)-\psi_Y(t).
\end{eqnarray*}
Hence $\rho_t(x,y)=p(x)\cdot\mu(y)^{1-t}\nu(y)^t e^{-\psi_Y(t)}$, which is a product measure on $X\times Y$, of the form $\P\otimes Q_t$, with $Q_t(y):=\mu(y)^{1-t}\nu(y)^t e^{-\psi_Y(t)}\in\Delta(Y)$, which shows $e$-convexity.  

From $m_t(x,y)=(1-t)p(x)q_1(y)+tp(x)q_2(y)=p(x)((1-t)q_1(y)+tq_2(y))$ we obtain $m$-convexity.

In summary, $\mathscr{E}_{0}(\P)$ is both $e$- and $m$-convex.
\end{proof}

\begin{figure}
\centering
\begin{subfigure}{0.49\textwidth}
\centering
\includegraphics[width = \textwidth]{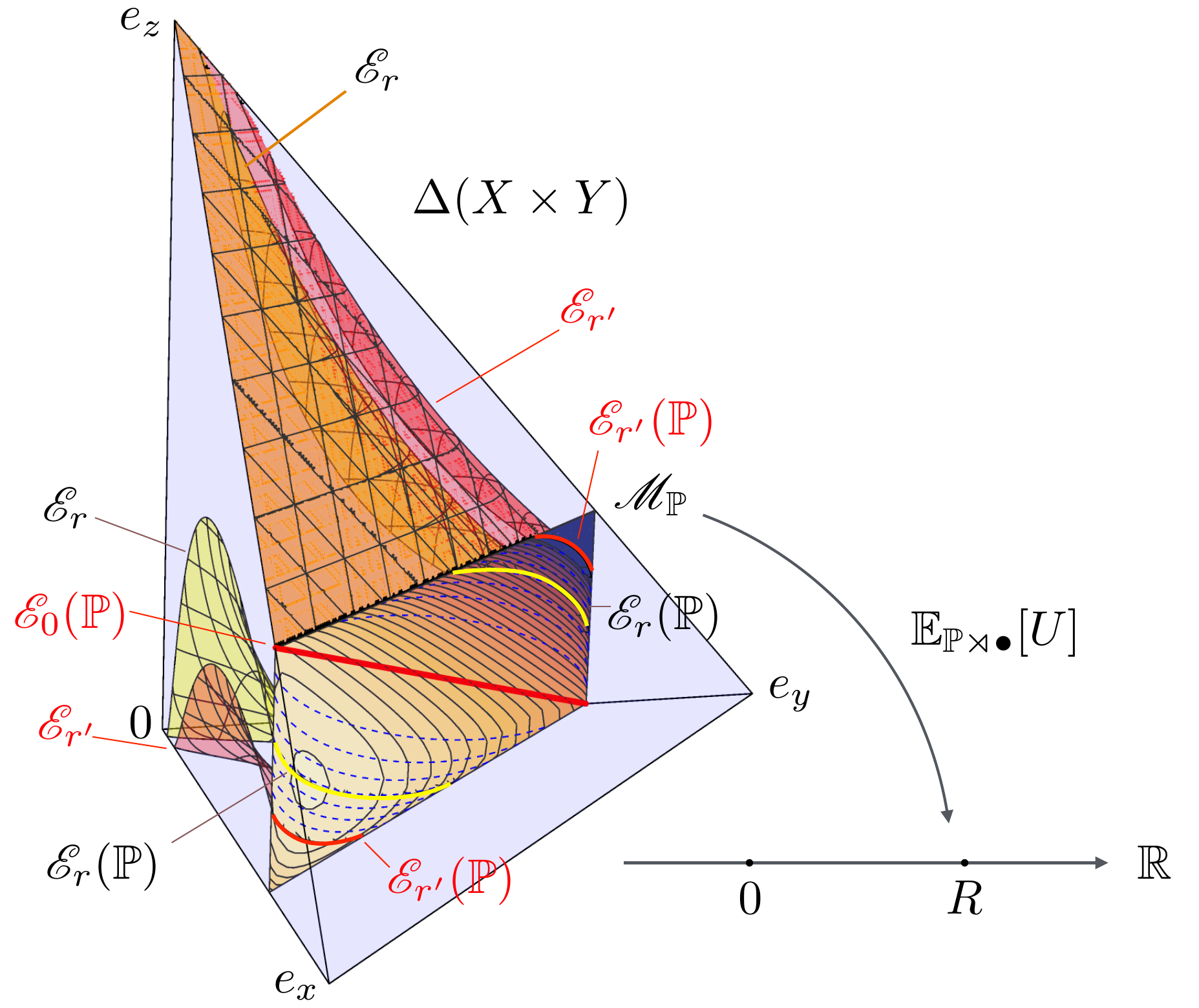}
\caption{Information foliatation}
\label{fig:Information foliatation}
\end{subfigure}
\begin{subfigure}{0.46\textwidth}
\centering
\includegraphics[width = \textwidth]{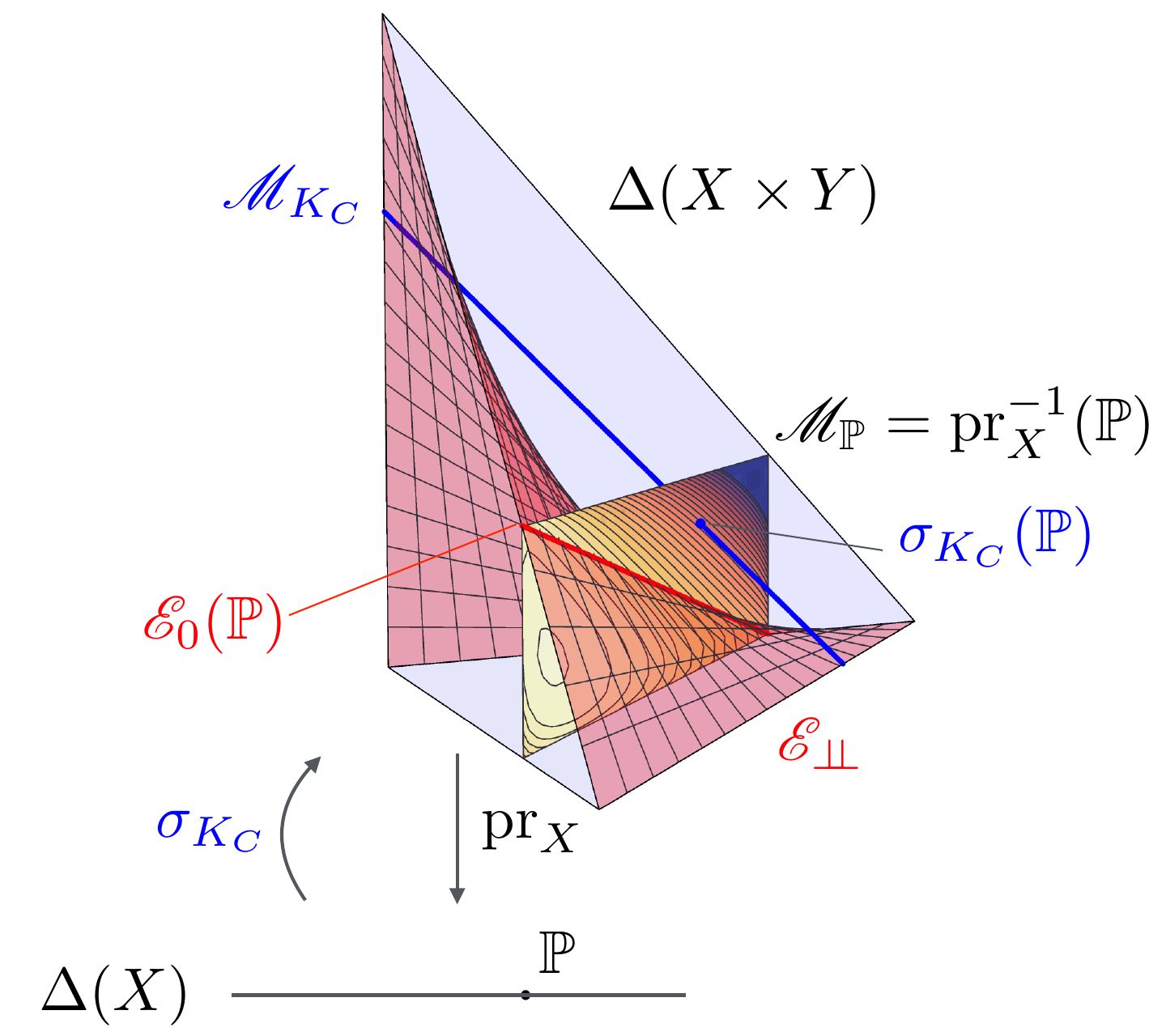}
\caption{Rate distortion}
\label{fig:Rate distortion}
\end{subfigure}
\caption{Left panel~\ref{fig:Information foliatation}: Submanifolds $\mathscr{M}_{\P}$,  ${\color{red}\mathscr{E}_r}$ and ${\color{red}\mathscr{E}_{r'}}$, $r<r'$, of $\Delta(X\times Y)$. Black contour lines surround the coloured areas corresponding to the different values of the free energy difference $F_{\beta}:=\E_{\pi}[U]-\frac{1}{\beta}I(\pi)$. 
The {\color{blue} blue dashed} lines show contours of constant mutual information. Right panel~\ref{fig:Rate distortion}: Submanifolds $\mathscr{M}_{\P}$ (fibre over $\P$),  ${\color{red}\mathscr{E}_{\perp\!\!\!\perp}}$ (exponential family of product measures) and {\color{red}$\mathscr{E}_0(\P)$} of $\Delta(X\times Y)$. ${\color{red}\mathscr{E}_0(\P)}=\mathscr{M}_{\P}\cap{\color{red}\mathscr{E}_{\perp\!\!\!\perp}}$. Section ${\color{blue}\sigma_{K_C}}:\Delta(X)\rightarrow\Delta(X\times Y)$, $\mu\rightarrow\mu\rtimes K_C$ induced by channel $K_C$, with image ${\color{blue}\mathscr{M}_{K_C}}$, and ${\color{blue}\sigma_{K_C}(\P)}:=\mathscr{M}_{\P}\cap{\color{blue}\mathscr{M}_{K_C}}$.}
\label{fig:combined1}
\end{figure}
\begin{figure}
\centering
\begin{subfigure}{0.33\textwidth}
\centering
\includegraphics[width = \textwidth]{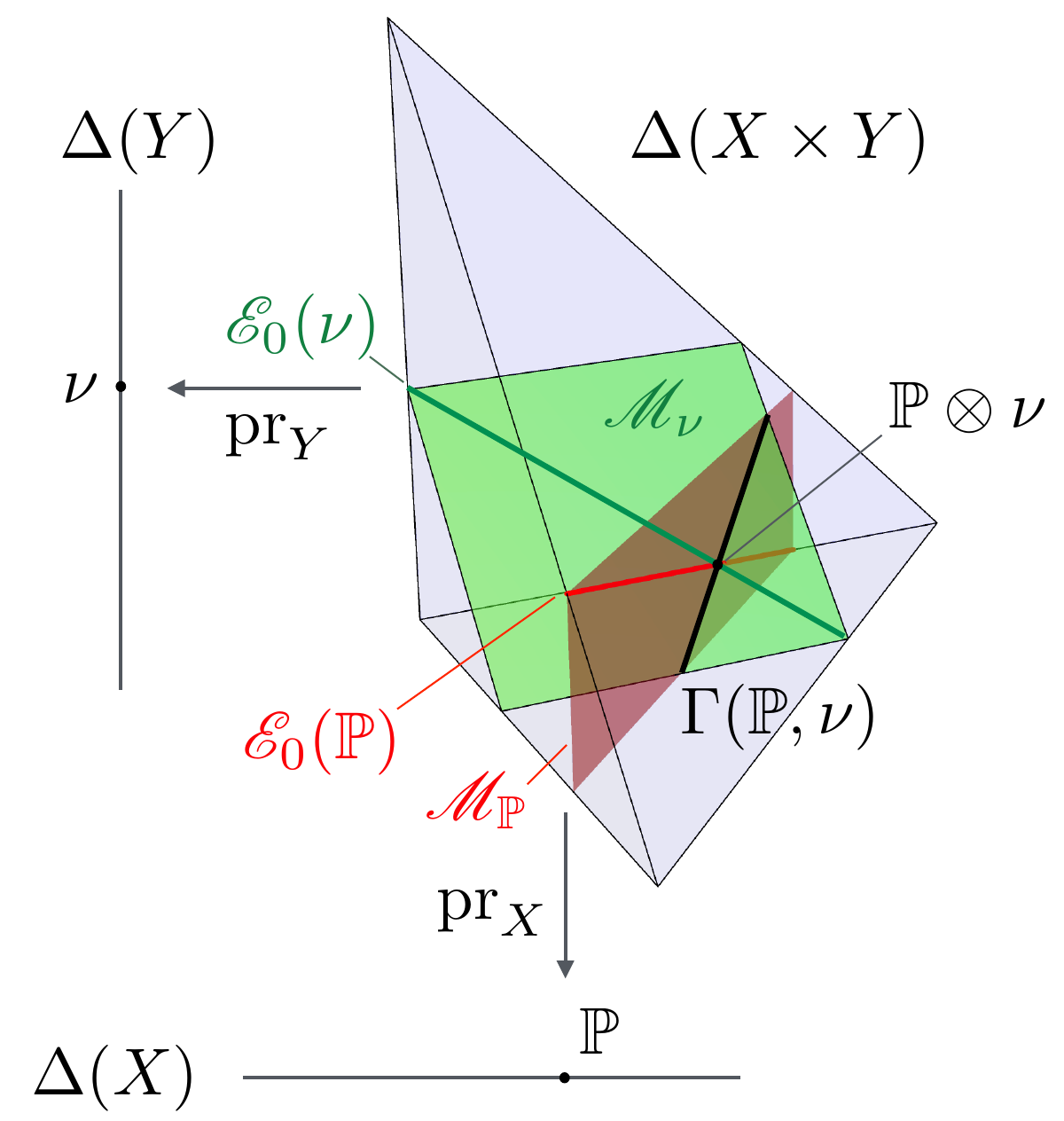}
\caption{Couplings}
\label{fig:couplings}
\end{subfigure}
\begin{subfigure}{0.46\textwidth}
\centering
\includegraphics[width = \textwidth]{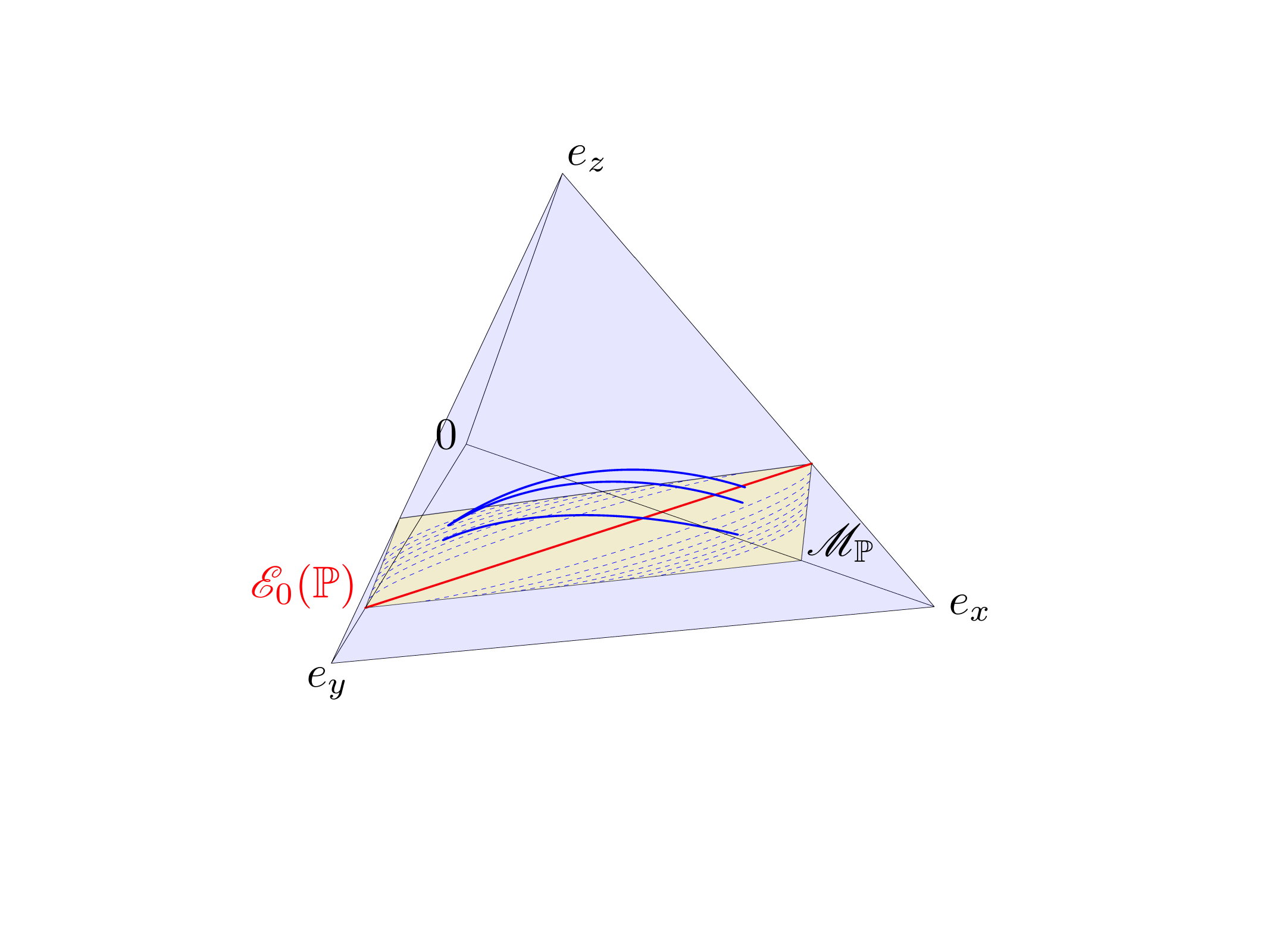}
\caption{$e$-geodesics}
\label{fig:e_geodesics}
\end{subfigure}
\caption{Set of product measures ${\color{red}\mathscr{E}_0(\P)}$ ({\color{red} red diagonal}) in $\mathscr{M}_{\P}$ and ${\color{teal}\mathscr{E}_0(\nu)}$ ({\color{teal} dark green diagonal}) in ${\color{teal}\mathscr{M}_{\nu}}$. Left panel~\ref{fig:couplings}: $m$-convex set of couplings of $\P$ and $\nu$  in $\Delta(X\times Y)$, given by the intersection $\Gamma(\P,\nu)={\color{teal} \mathscr{M}_{\nu}}\cap{\color{red}\mathscr{M}_{\P}}$, and ${\color{red}\mathscr{E}_0(\P)}\cap{\color{teal}\mathscr{E}_0(\nu)}=\P\otimes\nu$. Right panel~\ref{fig:e_geodesics}: Rear view of the probability simplex $\Delta(2\times 2)$. The {\color{blue} blue dashed} lines show contours of constant mutual information. The thick {\color{blue} blue curves} represent three different {\color{blue}$e$-geodesics} in $\Delta(2\times 2)$ that start and end in $\mathscr{M}_{\P}$ but are not otherwise included.}
\label{fig:combined2}
\end{figure}

\begin{lem}
\label{prop:properties}
Let $(X,\mathscr{X})$ be a measurable space and $(Y,\mathscr{Y})$ a Borel space.\footnote{The `standard' or Borel property already follows from the fact that $Y$ is countable. Furthermore, we identify the space of finite probability distributions with subsets of some $\R^d$, i.e. we don't distinguish between projection maps and the induced push-forward on measures.} Then the following properties hold for all $\P\in\Delta^{\circ}(X)$ and $\nu\in\Delta(Y)$.
    \begin{enumerate}
    \item The fibre $\operatorname{pr}_X^{-1}(\P)\subset\Delta(X\times Y)$ over $\P$ is equal to $\mathscr{M}  _{\P}$. The projection $\operatorname{pr}_Y:\mathscr{M}_{\P}\rightarrow\Delta(Y)$ is surjective.
     \item $\mathscr{M}_{\P}$ is $m$-convex and compact.
     \item $\mathscr{M}_{K_C}$ is $m$-convex and transversal to $\mathscr{M}_{\P}$.
      \item The set of couplings $\Gamma(\P,\nu)$ is compact and $m$-convex, and it satisfies $\Gamma(\P,\nu)=\mathscr{M}_{\P}\cap\mathscr{M}_{\nu}$.
    \end{enumerate}
\end{lem}
\begin{proof}
We show only the first statement. Since $Y$ is finite, it is Borel. So the statement follows from the `Disintegration theorem'~\cite[Theorem 2.18]{ccinlar2011probability}.
\end{proof}

\begin{lem}
The set $C_r:=\{I(\P;K)\leq r\}$ is a compact convex subset of $\mathscr{M}_{\P}$.
\end{lem}
\begin{proof}
$\Delta(X\times Y)$ is a compact and convex subset of $\R^{|X||Y|}$ for any finite sets $X$ and $Y$. The rest follows from the convexity of $D_{\operatorname{KL}}(p||q)$ in pairs $(p,q)$~\cite[Thm. 2.7.2]{cover1991elements}, and that for fixed $p(x)$, $I(X,Y)$ is convex~\cite[Thm. 2.7.4]{cover1991elements}.
\end{proof}
For $\pi\in\Delta(X\times Y)$, a non-independent distribution, and $\rho\in\mathscr{E}_{\perp\!\!\!\perp}(X\times Y)$, an independent distribution, the solution to
\begin{equation}
\label{eq:mini_mutual_info}
\rho^*:=\argmin_{\rho\in\mathscr{E}_{\perp\!\!\!\perp}}D_{\operatorname{KL}}(\pi||\rho),
\end{equation}
is $\rho^*=\pi_X\otimes\pi_Y$, i.e. the product of the $X$- and $Y$-marginal of $\pi$, cf.~\cite[p. 46]{amari2016information} and~\cite[Lemma 13.8.1]{cover1991elements}.
Geometrically, $\rho^*$ corresponds to the point where the $m$-geodesic starting at $\pi$ orthogonally intersects $\mathscr{E}_{\perp\!\!\!\perp}$, i.e. the $m$-projection of $\pi$ onto $\mathscr{E}_{\perp\!\!\!\perp}$. Thus, the mutual information $I(\pi)$ quantifies the deviation of $\pi$ from independence~\cite[p. 46]{amari2016information}.

In the theory of rate distortion~(\cite{Berger1971,cover1991elements,hino2024geometry}), the aim is to maximise the distance between $\mathscr{M}_{K_C}$ and $\mathscr{E}_{\perp\!\!\!\perp}$, cf.~Figure~\ref{fig:Rate distortion}.

\subsection{Legendre transform}
\begin{lem}
Let $\P\in\Delta^{\circ}(X)$, $X$ and $Y$ be finite, and, considered as random variables, $(X,Y)\sim\pi(x,y)=p(x)k(x,y)$. We have:
\begin{enumerate}
\item $I\big|\in\mathscr{C}^{\infty}(\Delta^{\circ}(X\times Y),\R_+)$, i.e. the mutual information is a smooth function of its arguments $\pi(x,y)$. 
\item $I\big|:\mathscr{M}^{\circ}_{\P}\rightarrow\R_+$ is a strictly convex $\mathscr{C}^{\infty}$-function on the interior $\mathscr{M}_{\P}^{\circ}$ of $\mathscr{M}_{\P}$, i.e. for $\pi_1,\pi_2\in\mathscr{M}^{\circ}_{\P}$, $\pi_i\notin\mathscr{E}_{\P}(0)$, $i=1$ or $i=2$, and $\lambda\in[0,1]$, the strict inequality
\begin{equation}
I(\lambda (\pi_1)+(1-\lambda)\pi_2)<\lambda I(\pi_1)+(1-\lambda)I(\pi_2),
\end{equation}
holds.
\item The restriction $I\big|:\mathscr{M}_{K_C}\rightarrow\R_+$ is concave.
\end{enumerate}
\end{lem}

\begin{proof}

The first statement is shown as follows. 
For $\pi(x,y)\in\Delta^{\circ}(X\times Y)$, $\pi_X(x)=\sum_y\pi(x,y)$ and $\pi_Y(y)=\sum_x\pi(x,y)$ are the marginals. The claim follows from
\begin{equation*}
\begin{split}
\frac{\partial}{\partial\pi(x,y)}I(\pi)&=\frac{\partial}{\partial\pi(x,y)}\sum_{x',y'}\pi(x',y')\ln\frac{\pi(x',y')}{\sum_{\nu}\pi(x',\nu)\sum_{\mu}\pi(\mu,y')}\\
&=\ln\frac{\pi(x,y)}{\pi_X(x)\pi_Y(y)}-\frac{\pi(x,y)}{\pi_X(x)}-\frac{\pi(x,y)}{\pi_Y(y)}+1.\\
\end{split}
\end{equation*}

Note, $\mathscr{M}_{\P}$ is convex, and thus its interior $\mathscr{M}_{\P}^{\circ}$ is a convex domain for the restriction $I\big|$.

The smoothness of the restricted mutual information follows from the first statement, by restriction to a relatively open set.

Next, consider two different kernels $K_1(x,y)$ and $K_2(x,y)$. The corresponding joint distributions are $\pi_1(x,y)=p(x)k_1(x,y)$ and $\pi_1(x,y)=p(x)k_1(x,y)$, and their respective $X$- and $Y$-marginals are $p(x)$ and $q_i(y):=\sum_xp(x)k_i(x,y)$ for $i=1,2$.

Define $\pi_{\lambda}(x,y):=\lambda\pi_1(x,y)+(1-\lambda)\pi_2(x,y)$ and $\rho_{\lambda}(x,y):=\lambda p(x)q_1(y)+(1-\lambda)p(x)q_2(y)$ as the product of the marginal values.

Then  
$I(\pi_{\lambda})=D_{\operatorname{KL}}(\pi_{\lambda}||\rho_{\lambda})=\sum_{x,y}\pi_{\lambda}(x,y)\ln\frac{\pi_{\lambda}(x,y)}{\rho_{\lambda}(x,y)}$, which is a smooth function of $\lambda\in(0,1)$.
Strict convexity follows from 
$$
\frac{d^2}{d\lambda^2}D_{\operatorname{KL}}(\pi_{\lambda}||\rho_{\lambda})=\sum_{x,y}\underbrace{\frac{d^2}{d\lambda^2}\left(\pi_{\lambda}(x,y)\ln\frac{\pi_{\lambda}(x,y)}{\rho_{\lambda}(x,y)}\right)}_{(*)>0}>0,
$$
given that $(*)$ is positive for every pair $(x,y)$. The latter follows from 
$$
\frac{(\pi_2 \rho_1 - \pi_1 \rho_2)^2}{\left(\lambda\pi_1 +(1-\lambda)\pi_2\right) \left(\lambda \rho_1 + (1-\lambda)\rho_2\right)^2}>0,
$$
which is true for all $\lambda\in(0,1)$.

The third statement is a direct consequence of~\cite[Theorem 2.7.4]{cover1991elements}. 
\end{proof}
In particular, the level sets $\mathscr{E}_{\P}(r)$, $r>0$, define strictly convex subsets of $\mathscr{M}^{\circ}_{\P}$, and therefore, 
the supporting hyperplane to $\mathscr{E}_{\P}(r)$ exists at $\pi_0$. For $\pi,\pi_0\in\boldsymbol{\operatorname{dom}}(I)$, $\pi\neq\pi_0$, it satisfies the strict first-order condition 
\begin{equation}
\label{eq:first_order}
I(\pi)>I(\pi_0)+\langle\nabla I(\pi_0),(\pi-\pi_0)\rangle,
\end{equation}
where $\nabla I$ is the usual gradient; cf.~\cite[p.~70]{boyd2004convex}.
The Bregman divergence can then be expressed as 
$$
D_I[\pi:\pi_0]:=I(\pi)-I(\pi_0)-\langle\nabla I(\pi_0),(\pi-\pi_0)\rangle,
$$
where $\pi^*:=\nabla I(\pi)$ is equal to the Legendre transform of $\pi$,~\cite[p. 14]{amari2016information}.

\section{A multiplier-robust control problem}
Let $(Y,2^Y)$ be a discrete measurable space and $U:Y\rightarrow\R$ be a measurable utility function. Let $u_i:=U(i)$, $i^*_{\max}:=\argmax_{i\in Y}u_i$ and $i^*_{\min}:=\argmin_{i\in Y}u_i$; assuming that $\max$ and $\min$ are unique.
Fix both $q\in\Delta^{\circ}(Y)$, the prior, and $\beta>0$. The parameter $\beta$ is called the inverse temperature in thermodynamics, i.e. $1/\beta$ is the physical temperature. 

The model of resource-bounded rationality that we will consider is given by the following multiplier-robust control problem~\cite{mattsson2002probabilistic,ortega2013thermodynamics,genewein2015bounded}:
\begin{equation}
\label{eq:max}
U_{\beta}^*:=\max_{p\in\Delta_n}\left[\E_p[U]-\frac{1}{\beta}D_{KL}(p||q)\right]\quad \text{and} \quad p^*:=\argmax_{p\in\Delta_n}\left[\E_p[U]-\frac{1}{\beta}D_{KL}(p||q)\right].
\end{equation}
Then the optimal solution $U_{\beta}^*$ is called the (negative) free energy difference~\cite{ortega2013thermodynamics,genewein2015bounded}, or the certainty equivalent~\cite{ortega2016human}. The partition function is defined as
\begin{equation}
\label{eq:partition_sum}
Z_{\beta}:=\sum_{i\in Y} q_i e^{\beta u_i}.
\end{equation}
The free energy is given by
\begin{equation}
\label{eq:partition function} 
\beta\mapsto\ln Z_{\beta},
\end{equation}
which is known in thermodynamics as the Massieu function~\cite{naudts2009q}. The Boltzmann-Gibbs distribution (cf.~\cite{naudts2009q,amari2016information,ay2017information}) is defined as
\begin{equation}
\label{eq:Boltzmann_Gibbs}
p_{\beta}(i):=\frac{1}{Z_{\beta}}q_ie^{\beta u_i},\quad i\in Y.
\end{equation}
The Boltzmann distribution, which plays an important role in the multinomial logistic choice model, is obtained for a uniform prior probability $q_i\equiv1/(n+1)$, cf.~\cite{Mcfadden2012,mattsson2002probabilistic}. 
In addition, we have
\begin{equation}
\label{eq:first_cumulant}
\frac{d}{d\beta}\ln Z_{\beta}=\frac{Z'_{\beta}}{Z_{\beta}}=\E_{p_{\beta}}[U],
\end{equation}
which is the first cumulant or the expectation value of $U$, and 
\begin{equation}
\label{eq:second_cumulant}
\begin{split}
\frac{d^2}{d\beta^2}\ln Z_{\beta}&=(\frac{Z'_{\beta}}{Z_{\beta}})'=\frac{Z''_{\beta}}{Z_{\beta}}-(\frac{Z'_{\beta}}{Z_{\beta}})^2=\E_{p_{\beta}}[U^2]-\E_{p_{\beta}}[U]^2\\
&=\underbrace{\E_{p_{\beta}}[(U-\E_{p_{\beta}}[U])^2]}_{=\operatorname{Var}_{p_{\beta}}(U)}>0,
\end{split}
\end{equation} 
which is the second cumulant or the variance of $U$. Hence, by~$(\ref{eq:second_cumulant})$, the free energy is strictly convex if $U$ is not constant.
\begin{prop}
The free energy difference $U^*({\beta}):=\frac{1}{\beta}\ln Z_{\beta}$ is a strictly increasing sigmoidal function. The following asymptotic bounds hold:
\begin{align*}
\lim_{\beta\to0}\frac{1}{\beta}Z_{\beta}&=\E_q[U]&\lim_{\beta\to0} p_{\beta}&=q,\quad\text{(high-temperature regime)}\\
\lim_{\beta\to+\infty}\frac{1}{\beta}Z_{\beta}&=u_{i^*_{\max}}&\lim_{\beta\to+\infty} p_{\beta}&=\delta_{i^*_{\max}},\quad\text{(low-temperature regime)}\\
\lim_{\beta\to-\infty}\frac{1}{\beta}Z_{\beta}&=u_{i^*_{\min}} &\lim_{\beta\to-\infty}p_{\beta}&=\delta_{i^*_{\min}},
\end{align*}
where $u_{i^*_{\max}}$ is the maximum, $u_{i^*_{\min}}$ is the minimum, and $\delta_i$ is the Dirac measure at the vertex $i$.
\end{prop}
\begin{proof}\footnote{I would like to thank Takuya Murayama for kindly providing me  the ansatz for the proof.}
Let us show the strict monotonicity of $U^*({\beta})$. Consider the case $\beta>0$. We have 
\begin{equation}
\frac{1}{\beta}\ln Z_{\beta}=\frac{1}{\beta}\int_0^{\beta}\frac{d}{dx}\left(\ln Z_x\right)dx=\frac{1}{\beta}\int_0^{\beta}\E_{p^*_x}[U]\,dx.
\end{equation}
By changing the variables as $x=\beta t$, one obtains
$$
\frac{1}{\beta} \ln Z_{\beta}=\int_0^1\E_{p^*_{\beta t}}[U] dt.
$$
For each fixed $t\in(0,1]$, the function $f_t(\beta):=\E_{p^*_{t\beta}}[U]$ is strictly increasing in $\beta$, as $\frac{d}{d\beta}f_t(\beta)=t\cdot\operatorname{Var}_{p^*_{t\beta}}(U)>0$ by~(\ref{eq:second_cumulant}), if not all $u_i$ are equal. 

The case of $\beta<0$, is treated by considering $\frac{1}{\beta}\int_{\beta}^0\frac{d}{dx}(\ln Z_{\beta})dx$.
\end{proof}
\begin{prop}[Free energy difference~\cite{mattsson2002probabilistic,ortega2013thermodynamics,genewein2015bounded}]
\label{prop:one}
For $q\in\Delta^{\circ}(Y)$, $\beta\in\R_+$, the resulting net utility $U_{\beta}^*$ is given by the free energy difference, i.e.
\begin{equation*}
U_{\beta}^*=\frac{1}{\beta}\ln Z_{\beta},
\end{equation*}
and the optimal distribution $p^*_{\beta}$, is given by the Boltzmann-Gibbs distribution
\begin{equation*}
p_{\beta}^*(i)=\frac{1}{Z_{\beta}}q_i e^{\beta u_i},\,\,i\in Y.
\end{equation*}
\end{prop}
Note that for $p=q$ the cost term in~(\ref{eq:max}) is zero for all $\beta$, but for $0 < \beta < \infty$ the optimal probability differs from the prior. Despite the positive net cost, the expected net utility is greater because the chosen probability better fits the structure of the utility function.

\begin{cor}
The expected utility, $\E_{p^*_{\beta}}[U]$, is strictly increasing for $\beta\in\R_+$, and it satisfies
$$
\E_q[U]< U^*_{\beta}<\E_{p^*_{\beta}}[U]<u_{i^*_{\max}}.
$$
\end{cor}

\subsection{Geometry of the solution: geodesics}
Let both $q\in\Delta^{\circ}(Y)$ and $r\geq0$ be fixed. We consider the optimisation problem
\begin{eqnarray}
\label{eq:constrained_opti}
\max_{p\in\Delta(Y)}\E_p[U]\\\nonumber
\text{s.t.}\quad D_{\operatorname{KL}}(p||q)\leq r. 
\end{eqnarray}
An application of the Karush-Kuhn-Tucker method~\cite[Kapitel~20]{pampel2017arbeitsbuch} gives the Lagrange function\footnote{Note that the use of $\beta$ to denote one of the Lagrange multipliers is, at this point, merely a matter of notation.}
\begin{equation}
L(p,\lambda,\xi,\beta)=\sum_{i=0}^n p_i u_i+\lambda\left(1-\sum_{i=0}^n p_i\right)+\sum_{i=0}^n\xi_i p_i+\beta\left(r-\sum_{i=0}^np_i\ln\frac{p_i}{q_i}\right),
\end{equation}
with $\xi_i\geq0$ and $\beta\geq0$, as required for a maximum.

The condition $\nabla_p L=0$, gives the following system of equations for $i\in[n]=Y$:\begin{eqnarray}
\label{eq:lagr1}
     u_i-\lambda+\xi_i-\beta(\ln\frac{p_i}{q_i}+1)=0,&&\\
     \label{eq:lagr2}
     1-p_1-\dots-p_n=0,&&\\
     \label{eq:lagr3}
     \xi_i p_i=0,&&\\\label{eq:Div_KL_r}
     \label{eq:lagr4}
     \beta (r-D_{\operatorname{KL}}(p||q))=0.&&
\end{eqnarray}
In order for the solution to satisfy $p\in\Delta^{\circ}(Y)$, it is necessary that  $\xi_i=0$ for all $i\in Y$.

Condition~(\ref{eq:lagr1}) gives $q_ie^{-1}e^{\beta(u_i-\lambda)}=p_i$, and by summing over $i$ and taking~(\ref{eq:lagr2}) into account, we get $\sum_i q_ie^{-1}e^{\beta(u_i-\lambda)}=1$. This gives $\beta^{-1}(\ln Z_{\beta}-1)=\lambda$, where $Z_{\beta}$ is defined in~(\ref{eq:partition_sum}). Then for $\beta$ fixed, and with the substitution $\beta\to\frac{1}{\beta}$, we get
\begin{equation}
\label{eq:solution_Boltzmann}
p_i=q_i\frac{1}{Z_{\beta}}e^{\beta u_i},
\end{equation}
which is the Gibbs distribution.

The condition~(\ref{eq:lagr4}) requires a distinction of cases: for $\beta=0$ we get from~(\ref{eq:solution_Boltzmann}) that $p\equiv q$. For $\beta\neq0$, 
the optimal value of the Lagrange multiplier $\beta^*=\beta^*(r)$ is determined from the condition $r=D_{\operatorname{KL}}(p||q)$. The solution is then given by the first-order differential equation
\begin{equation}
\label{eq:bregman}
r=\frac{1}{Z_{\beta}}\sum_{i=0}^nq_ie^{\beta u_i}\ln\frac{e^{\beta u_i}}{Z_{\beta}}=\beta\cdot\E_{p_{\beta}^*}[U]-\ln Z_{\beta}=\beta\frac{d}{d\beta}\ln Z_{\beta}-\ln Z_{\beta},
\end{equation}
where strict monotonicity of $U^*_{\beta}$ implies that for $r>0$, $\beta^*>0$ must hold.

\begin{lem}
\label{eq:r_beta_relation}
The function $r=r(\beta)$, defined by~(\ref{eq:bregman}), is strictly increasing for $\beta\in\R_+$ (cf. figure below). 
It satisfies $r(0)=0$ and $\lim_{\beta\to+\infty} r(\beta)=r_{\max}:=-\ln q_{i^*_{\max}}$, $i^*_{\max}:=\argmax_{i\in Y} u_i$. 
Hence the inverse function $r\mapsto\beta(r)\in\R_+$  exists for all $r\in[0,r_{\max}]$, and is necessarily strictly increasing.
\end{lem}
\begin{equation*}
\qquad\qquad\qquad\qquad\qquad\includegraphics[width=0.3\linewidth]{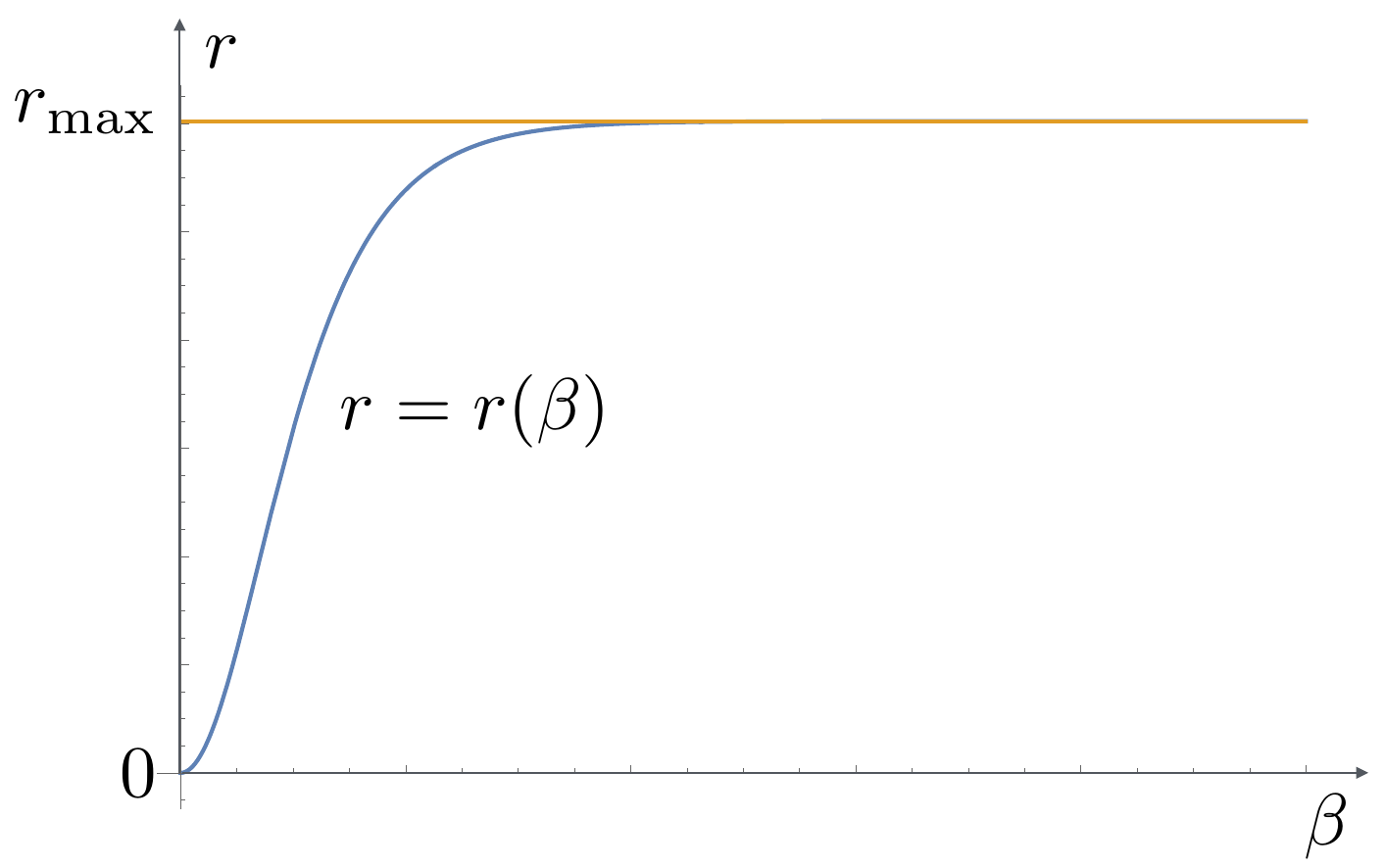} 
\end{equation*}
\begin{proof}
From (\ref{eq:bregman}) we have $r(\beta)=\beta\frac{d}{d\beta}\ln Z_{\beta}-\ln Z_{\beta}$, and so for $\beta>0$ we get
$$
r'(\beta)=\frac{d}{d\beta}\ln Z_{\beta}+\beta\frac{d^2}{d\beta^2}\ln Z_{\beta}-\frac{d}{d\beta}\ln Z_{\beta}=\beta\cdot\operatorname{Var}_{p^*_{\beta}}(U)>0\quad(\text{by~(\ref{eq:second_cumulant}})). 
$$
\end{proof}
Let us make the following remark. The relation between the formulations in~(\ref{eq:max}) and in~(\ref{eq:constrained_opti}) follows from Lemma~\ref{eq:r_beta_relation} and is given by the correspondence $1/\beta(r) \leftrightarrow r$.
\begin{prop}
The solution of~(\ref{eq:max}) gives an $e$-geodesic $\gamma^{(e)}:\R\rightarrow\Delta(Y)$ with initial conditions $\gamma^{e}(0)=q$, $\dot{\gamma}^{(e)}(0)\in T_q\Delta^{\circ}(Y)$ for the tangent vector, whose components are   
\begin{equation}
 \dot{\gamma}^{(e)}_i(0)=q_i(U(i)-\E_q[U]),\quad i\in Y.
\end{equation}
The geodesic has the following asymptotic bounds:
\begin{eqnarray}
    \lim_{t\to+\infty}\gamma^{(e)}(t)&=&\delta_{\argmax_iU(i)},\\
    \lim_{t\to-\infty}\gamma^{(e)}(t)&=&\delta_{\argmin_iU(i)}.
\end{eqnarray}

\end{prop}

\begin{figure}
    \centering
    \includegraphics[width=0.5\linewidth]{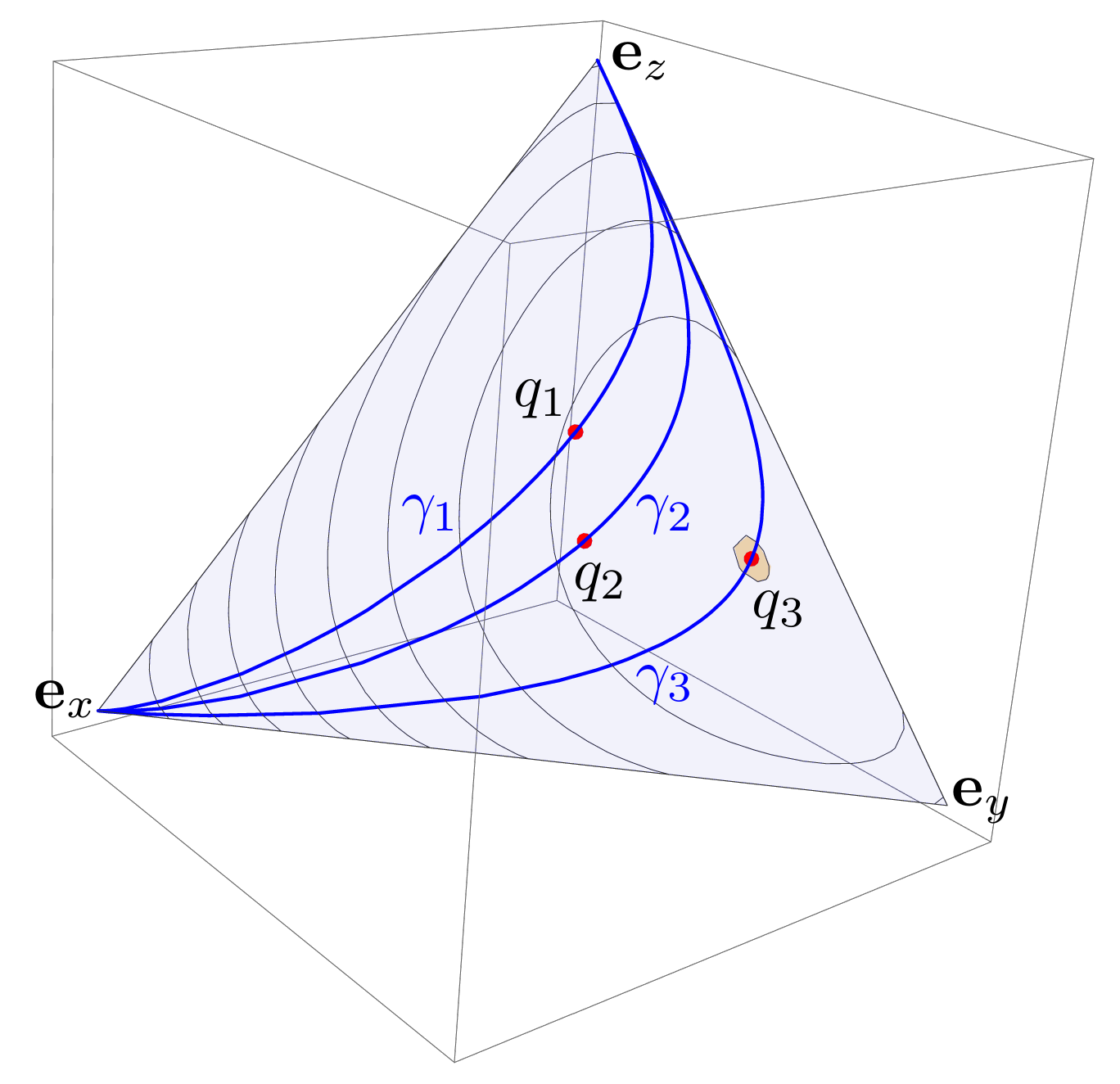}
    \caption{The probability simplex with priors $q_1$, $q_2$ and $q_3$ and $e$ geodesics ${\color{blue} \gamma_i}$ through $q_i$ ({\color{blue} blue curves}), parameterised by $\beta$. The contours show different KL spheres in $\Delta_2$ with different radii, all centred on point $q_3$.}
   \label{fig:Bregman_Prob_simplex}
\end{figure}
The explicit solution is obtained by solving the second-order differential equation\footnote{For simplicity, we omit the superscript $(e)$.}
\begin{equation}
\label{eq:second_order_diff}
    \ddot{\gamma}-\frac{\dot{\gamma}^2}{\gamma}+\gamma\sum_{i\in Y}\frac{\dot{\gamma}^2_i}{\gamma_i}=0,
\end{equation}
with initial conditions $\gamma(0)=q\in\Delta^{\circ}(Y)$ and $\dot{\gamma}(0)=v\in T_{q}\Delta^{\circ}(Y)$. The solution is given in components by (cf.~\cite[p. 46]{ay2017information})
\begin{equation}
\label{eq:geodesic}
t\mapsto\frac{q_ie^{t\frac{v_i}{q_i}}}{\sum_{j\in[n]}q_je^{t\frac{v_j}{q_j}}},\quad i\in Y.
\end{equation}
Then $\gamma^{(e)}(t):=q\frac{1}{Z_t}e^{tU}$, with 
\begin{equation}
    v_i=\dot{\gamma}_i^{(e)}(0)=\frac{d}{dt}\Big|_{t=0}q_i\frac{e^{tu_i}}{Z_t}=q_i\left(u_i-\E_q[U]\right),\quad i\in Y,
\end{equation}
as the tangent vector at the origin, does indeed satisfy~(\ref{eq:second_order_diff}), as a calculation shows. 
We also have $\sum_{i\in Y} v_i=\sum_{i\in Y}q_iu_i-\E_q[U]\cdot\sum_{i\in Y}q_i=0$, as required for a tangent vector.  
The Boltzmann-Gibbs distribution~(\ref{eq:Boltzmann_Gibbs}) can in fact be interpreted as the gradient flow of the expected value of the utility function $U$, cf.~\cite[Example 3.3, p. 229]{pistone2020information}. 

For $X$ and $Y$ discrete, let $U(x,y)$ be a state-dependent utility function as in~(\ref{eq:utility}) and $\kappa\in\mathfrak{K}(X,Y)$ be a Markov kernel representing a collection of state-dependent a priori measures.  
Then the corresponding solution is given by a one-parameter family of Markov kernels $K^*_{\beta}$, with components 
\begin{equation}
\label{eq:individual_free_energy}
k^*_{\beta}(x,y)=\kappa(x,y)\frac{1}{Z_{\beta}(x)}e^{\beta U(x,y)},
\end{equation}
where $Z_{\beta}(x):=\sum_{y\in Y}\kappa(x,y)e^{\beta U(x,y)}$ is the state-dependent partition sum.

Geometrically, $\beta^*(r)$ is given by the intersection of the geodesic through $q$ and the Bregman sphere for increasing values of $r$, cf. Figure~\ref{fig:Bregman_Prob_simplex}.
Changing the prior $q$ while keeping the utility fixed gives different geodesics that still converge to the vertices where the utility reaches its maximum and minimum.
If the utility function is changed but the prior is kept, the geodesic still passes through $q$, but the vertices may be different.\footnote{For example, consider the following cases: $\left(\begin{array}{cc}1 & 0 \\0 & 0\end{array}\right)$, $\left(\begin{array}{cc}1 & 0 \\1 & 0\end{array}\right)$, $\left(\begin{array}{cc}1 & 1 \\0 & 0\end{array}\right)$, $\left(\begin{array}{cc}1 & 0 \\0 & 1\end{array}\right)$}

\section{A constraint robust-control problem}
\label{sec:robust}
We discuss a generalisation of~\cite{tishby2000information} and~\cite[p. 4]{genewein2015bounded} from a geometrical point of view.
Let $X$ and $\tilde{Y}$ be finite sets, $(X, 2^X,\P)$ a finite probability space with source $\P=(p(x))_{x\in X}$ and $U:X\times\tilde{Y}\to\R$ a state-dependent utility function as in~(\ref{eq:utility}). 
Given a source probability, the question of which Markov prior ${\color{red}\kappa}$ maximises the individual differences in free energy on average, as in~(\ref{eq:max},\ref{eq:individual_free_energy}), leads to the following optimisation problem
\begin{equation}
\label{eq:inf_inequality}
\argmax_{\color{red}\kappa}\left(\sum_{x\in X}p(x)\cdot\left(\max_{K_x\in\Delta(\operatorname{supp}({\color{red}\kappa_x}))}\left(\E_{K_x}[U_x]- \frac{1}{\beta}D_{\operatorname{KL}}(K_x||{\color{red}\kappa_x})\right)\right)\right).
\end{equation}
First, for~(\ref{eq:inf_inequality}) to be well defined at all, {\color{red} $\operatorname{supp}(K_x)\subset\operatorname{supp}(\kappa_x)$} must hold for all $x\in X$ and all kernels. 
Next we assume that the support of all kernels to be considered is given by a rule, i.e. $\operatorname{supp}(\kappa_x)=\mathfrak{O}_x\subset\tilde{Y}$ for all $x\in X$ and $\kappa$. 
Then a uniform domain $Y\subset\tilde{Y}$ can be obtained for the optimisation problem, e.g. by setting $Y: =\bigcap_{x\in X}\mathfrak{O}_x$ and considering only priors $\kappa\in\mathfrak{K}(X,Y)$ with $\operatorname{supp}(\kappa_x)=Y$ for all $x$ (written as $\operatorname{supp}(\kappa)=Y$ for short).
However, restricting the set of possible actions on the basis of legal or moral constraints leads from a purely utilitarian first-best problem to a second-best problem, cf.~Section~\ref{sec:legal_app} and~\cite[p. 818 et seq.]{Mas-Colell1995}. 

\subsection{Generic kernel}
For a generic prior $\kappa\in\mathfrak{K}(X,Y)$ with $\operatorname{supp}(\kappa)=Y$ and $K\in\mathfrak{K}(X,Y)$, considered as a parameter, set $\pi(x,y):=p(x)K(x,y)$. 
Then, instead of maximising~(\ref{eq:inf_inequality}), we can equivalently determine the kernel $\kappa$ which, for a fixed $K$, minimises the expected relative entropy $\sum_{x\in X} p(x) D_{\operatorname{KL}}(K(x,y)||\kappa_x(y))$, i.e. 
\begin{equation}
\label{eq:generic_kernel}
\kappa^*:=\argmin_{\kappa\in\mathfrak{K}(X,Y)}\sum_{x\in X,y\in Y}\pi(x,y)\ln\frac{\pi(x,y)}{p(x)\kappa(x,y)}=\argmin_{\kappa\in\mathfrak{K}(X,Y)} D_{\operatorname{KL}}(\P\rtimes K||\P\rtimes\kappa).
\end{equation}
Since we have ${\tilde{\pi}(x,y):=p(x)\kappa(x,y)},\pi(x,y)\in\mathscr{M}_{\P}$, (cf. ~Figure~\ref{fig:combined1}), the minimum in~(\ref{eq:generic_kernel}) is attained for $\tilde{\pi}=\pi$, i.e. for $\kappa^*=K$.

Analytically, the Lagrangian corresponding to~(\ref{eq:generic_kernel}) is
$$
L(\kappa(x,y),\lambda_x)=\sum_{x,y}\pi(x,y)\ln\frac{K(x,y)}{\kappa(x,y)}+\sum_x\lambda_x(1-\sum_y\kappa(x,y)).
$$
Taking derivatives w.r.t. $\partial/\partial\kappa(x_i,y_j)$, for a given $K(x,y)$, one gets
\begin{equation}
\label{eq:aux1}
\pi(x_i,y_j)=\lambda_{x_i}\kappa(x_i,y_j).
\end{equation}
Summing over $y$ and taking into account the relations $\sum_y\pi(x_i,y)=p(x_i)$ and $\sum_y\kappa(x_i,y)=1$ gives $p(x_i)=\lambda_{x_i}$. Plugging this result into~(\ref{eq:aux1}) we obtain
$$
\pi(x_i,y_j)=p(x_i)\kappa(x_i,y_j).
$$
This shows that for all $x_i,y_j$ we must have $\kappa(x_i,y_j)=K(x_i,y_j)$, i.e. $\kappa^*=K$. 
The remaining inner optimisation in~(\ref{eq:inf_inequality}), which is independent of $\beta$, reduces to 
\begin{equation}
\label{eq:opt_inner_general}
K^*:=\argmax_{K\in\mathfrak{K}(X,Y)}\E_{\P\rtimes K}[U].
\end{equation}
The $\beta$-independent solution of~(\ref{eq:opt_inner_general}) is given for all $x\in X$ by
\begin{equation}
\label{eq:sol_opt_inner_general}
K^*_x=\delta_{xy^*(x)},\quad\text{for}~y^*(x):=\argmax_{y\in Y}U(x,y), 
\end{equation}
which corresponds to (\ref{eq:individual_free_energy}) for $\beta=+\infty$.
\subsection{Constant kernel}
Let $\kappa_q\in\mathfrak{K}(X,Y)$ be the constant kernel corresponding to $q\in\Delta^{\circ}(Y)$.
The outer maximisation over $\kappa_q$ in~(\ref{eq:inf_inequality}), and otherwise as discussed in~(\ref{eq:generic_kernel}), yields the condition 
\begin{equation}
\label{eq:constant_kernel}
q^*:=\argmin_{q\in\Delta^{\circ}(Y)} D_{\operatorname{KL}}(\P\rtimes K;\P\otimes q).
\end{equation}
Geometrically, the solution is given by the $m$-projection of $\pi:=\P\rtimes K$ onto the exponential submanifold $\mathscr{E}_{\P}(0)$, and it is equal to $\P\otimes K_*\P$, i.e. the solution $q^*$ is given by the $Y$-marginal $\pi_Y=K_*\P$ of $\P\rtimes K$, cf.~(\ref{eq:mini_mutual_info}).
In contrast to the generic situation, the solution is restricted to an exponential submanifold. 

Analytically, the outer maximisation over $\kappa_q$ in~(\ref{eq:inf_inequality}), and otherwise as above, gives for~(\ref{eq:constant_kernel}), with $q=(q(y))_{y\in Y}$,\footnote{We do not write out the condition $q_i\geq0$ explicitly.} the Lagrangian
$$
L(q,\lambda)=\sum_{x,y}\pi(x,y)\ln\frac{K(x,y)}{q(y)}+\lambda(1-\sum_yq(y)).
$$
Taking derivatives w.r.t. $\partial/\partial q(y_i)$, we get
$$
\underbrace{\sum_x \pi(x,y_i)}_{=\pi_Y(y_i)}=\lambda q(y_i),
$$
where $\pi_Y$ is the $Y$-marginal of $\pi(x,y)$. 
Summing over $y$, with $\sum_y\pi_Y(y)=1$ and $\sum_yq(y)=1$, gives $\lambda=1$. This implies that $q(y)=K_*\P(y)(=\pi_Y(y))$ for all $y\in Y$.
Hence, the optimal prior is equal to $q^*=K_*\P$, cf.~\cite[Csisz{\'{a}}r and Tusn{\'{a}}dy's Lemma 2]{genewein2015bounded,tishby2000information}.
Plugging $q^*=K_*\P$ into~(\ref{eq:inf_inequality}) gives $\E_{\P\rtimes K}[U]-\frac{1}{\beta}\sum_{x,y}\pi(x,y)\ln\frac{K(x,y)}{(K_*\P)(y)}$, which by~(\ref{eq:mutual_info}) gives $\E_{\P\rtimes K}[U]-\frac{1}{\beta}I(\P;K)$, cf.~\cite[Equation (7)]{genewein2015bounded}. 
The resulting inner maximisation problem is 
\begin{eqnarray}
\label{eq:constrained_information}
\bar{U}(R)&:=&\underset{\pi\in\mathscr{M}_{\P}}{\max}\,\, \E_{\pi}[U],\\\nonumber
&\text{s.t.}& I(\pi)\leq R, R\geq0.
\end{eqnarray}
The function $\bar{U}(R)$, as a function of $R\in\R_+$, is called the rate utility function.
In contrast to the rate distortion function~\cite[Section 2.3]{Berger1971} or~\cite[Chapter 13]{cover1991elements}, which is defined by the minimisation of distortion for a fixed rate, the rate utility function is defined by the maximisation of utility.

`Bauer's Maximum Principle'~\cite[p. 298]{guide2006infinite} states that every upper-semicontinuous convex function defined on a compact convex subset $C$ of a locally convex Hausdorff space, has a maximiser that is an extreme point.
An extreme point of $C$ is a point that does not lie on an open line segment connecting two points of $C$. 
\begin{lem}
\label{lem:existence_solution}
The maximisation problem~(\ref{eq:constrained_information}) has a solution, which is not necessarily unique. It is reached at an extreme point of the convex, compact set $C_R:=\{\pi:I(\pi)\leq R\}$.
\end{lem}
\begin{proof}
$\E_{\bullet}[U]$ is a continuous linear functional defined on the convex and compact set, $C_R$. Bauer's Maximum Principle implies that a maximum is attained at some extreme point of $C_R$.
\end{proof}
If the utility matrix $U$ has symmetries or degeneracies, then there are several $\pi\in C_R$ that maximise~(\ref{eq:constrained_information}), but in the generic case each point on the rate-utility curve is reached by a unique stochastic kernel (conditional probability). 
This is analogous to rate distortion theory, cf.~\cite[Theorem 2.4.2]{Berger1971}. 

Let $R_{\min}:=\min_{R\in\R_+}\left(\text{s.t.}~\max_{\pi\in\mathscr{M}_{\P}:I(\pi)\leq R}\E_{\pi}[U]=\bar{U}(R_{\max})\right)$ be the smallest value for which the utility rate function reaches its maximal value. By definition, it satisfies $0\leq R_{\min}\leq R_{\max}$.
\begin{lem}[Concavity of the rate-utility function]
\label{lem:utility_rate_function}
The rate-utility function $\bar{U}(R)$, defined in~(\ref{eq:constrained_information}), is a non-decreasing and concave function of $R\in[0,R_{\max}]$. It is strictly increasing on $[0,R_{\min}]$ and constant on $[R_{\min},R_{\max}]$.
\end{lem}

The subsequent proof is analogous to that of the convexity of the rate-distortion function~\cite[Lemma 13.4.1]{cover1991elements}.
\begin{proof}
$\bar{U}(R)$ is the maximum of the expectation over increasingly larger sets as $R$ increases. Thus $\bar{U}(R)$ is non-decreasing in $R$.

Let $(R_1,\bar{U}_1)$ and $(R_2,\bar{U}_2)$ be two rate utility pairs which lie on the rate-utility curve. Let the joint distributions that achieve these pairs be $\pi_1:=\P\rtimes K_1$ and $\pi_2:=\P\rtimes K_2$. Consider the distribution $\pi_{\lambda}:=\lambda\pi_1+(1-\lambda)\pi_2$, $\lambda\in[0,1]$. 

Mutual information is a convex function of the stochastic kernel (conditional distribution) for a fixed source (cf.~\cite[Theorem 2.7.4]{cover1991elements}) and hence by definition of the utility-rate function, 
\begin{equation}
\label{eq:mutual_info_convex}
I(\pi_{\lambda})\leq\lambda I(\pi_1)+(1-\lambda)I(\pi_2)\leq\lambda R_1+(1-\lambda) R_2=: R_{\lambda}.
\end{equation}
Thus, by the linearity of expected utility and by assumption, we have
\begin{equation}
\begin{split}
\lambda\bar{U}(R_1)+(1-\lambda)\bar{U}(R_2)&=\lambda\E_{\pi_1}[U]+(1-\lambda)\E_{\pi_2}[U]=\E_{\pi_{\lambda}}[U]\\
&\leq\max_{\pi:I(\pi)\leq R_{\lambda}}\E_{\pi}[U]=\bar{U}(R_{\lambda}),
\end{split}
\end{equation}
which proves that $\bar{U}[R]$ is a concave function of $R$, since $I(\pi_{\lambda})\leq R_{\lambda}$, by inequality~(\ref{eq:mutual_info_convex}).

The second statement is a consequence of the concavity of the function. Namely, a non-decreasing concave function on a compact interval can only be constant on a subinterval containing the right endpoint. 
\end{proof}

The Karush-Kuhn-Tucker method applied to~(\ref{eq:constrained_information}),
with $\E_{\pi}[U]=\sum_{i,j} p_i k_{ij} u_{ij}$, gives the Lagrange function
\begin{equation}
\label{eq:Lagrangian_utility_rate}
L(k_{ij},\lambda,\xi,\beta)=\E_{\pi}[U]+\sum_i\lambda_i\left(1-\sum_{j} k_{ij} \right)+\sum_{i,j}\xi_i k_{ij}+\beta\left(r-\sum_{ij} p_i k_{ij}\ln\frac{k_{ij}}{\sum_{\ell} p_{\ell} k_{\ell j}}\right),
\end{equation}
where $u_{ij}:=U(i,j)$, $K=(k_{ij})$, $\lambda=(\lambda_i)$, $\xi=(\xi_{ij})$. In addition, $\xi_{ij}\geq0$ and $\beta\geq0$ are necessary conditions for a maximum. Set $q_j:=\sum_i p_i k_{ij}$, which is the $j$th component of the $Y$-marginal of $\pi=(p_ik_{ij})_{i,j}$. 

The condition $\nabla_K L=0$ defines the stationary points, whose coordinates are given by the equations: 
\begin{eqnarray}
\label{eq:lagr1_kern}
\frac{\partial}{\partial k_{\mu\nu}}L(k_{ij},\xi,\lambda)&=&p_{\mu}u_{\mu\nu}-\beta p_{\mu}(\ln k_{\mu\nu}-\ln q_{\nu})+\xi_{\mu\nu}-\lambda_{\mu}=0,\\
     \label{eq:lagr2_kern}
     1-k_{i1}-\dots-k_{in}&=&0,\quad i=1,\dots,M,\\
     \label{eq:lagr3_kern}
     \xi_{ij} k_{ij}&=&0,\quad i=1,\dots, M, j=1,\dots,N,\\
     \label{eq:lagr4_kern}
     \beta(I(\P;K)-R)&=&0.
\end{eqnarray}
Note that solving the equations (\ref{eq:lagr1_kern})-(\ref{eq:lagr4_kern}) is analogous to solving the rate distortion equation~\cite[Sec. 2.5 and 2.6]{Berger1971}.

Since the matrix of $U$ is a free parameter in the model, we only consider the generic case. Furthermore, we do not consider all possible case distinctions for the determination of the Lagrange multipliers.
Define the auxiliary function
\begin{equation}
f_{\beta}(k_{ij}):=\sum_{ij}p_i k_{ij} u_{ij}-\frac{1}{\beta}\sum_{ij} p_i k_{ij}\ln\frac{k_{ij}}{\sum_{\ell} p_{\ell} k_{\ell j}}.
\end{equation}
Then
\begin{equation}
\begin{split}
\frac{\partial}{\partial k_{\mu\nu}}f_{\beta}(k_{ij})& =p_{\mu}u_{\mu\nu}-\frac{1}{\beta}\left(p_{\mu}\ln\frac{k_{\mu\nu}}{\sum_{\ell}p_lk_{\ell\nu}}+\sum_{ij} p_{i}k_{ij}\frac{\partial}{\partial k_{\mu\nu}}\ln\frac{k_{ij}}{\sum_{\ell}p_{\ell}k_{\ell j}}\right)\\
&=p_{\mu}u_{\mu\nu}-\frac{1}{\beta}\left(p_{\mu}\ln k_{\mu\nu}-p_{\mu}\ln q_{\nu}+0\right)\quad\text{(by identity~(\ref{eq:calculation}))}\\
&=p_{\mu}u_{\mu\nu}-\frac{p_{\mu}}{\beta}(\ln k_{\mu\nu}-\ln q_{\nu})
\end{split}
\end{equation}
where we use that
\begin{equation}
\label{eq:calculation}
\begin{split}
\sum_{ij}p_ik_{ij}\frac{\partial}{\partial k_{\mu\nu}}\left(\ln\frac{k_{ij}}{\sum_{\ell}p_{\ell}k_{\ell j}}\right)&=p_{\mu} k_{\mu\nu}\frac{1}{k_{\mu\nu}}-\sum_{ij}p_i k_{ij}\frac{\partial}{\partial k_{\mu\nu}}\ln\sum_{\ell}p_{\ell}k_{\ell j}\\
&=p_{\mu}-\sum_{i}p_ik_{i\nu}\frac{p_{\mu}}{q_{\nu}}=p_{\mu}-\frac{p_{\mu}}{q_{\nu}}\sum_i p_ik_{i\nu}=0.
\end{split}
\end{equation}

For $k_{ij}>0$, $\xi_{ij}=0$ must hold for all $i,j$.  We get from (\ref{eq:lagr1_kern}), after substituting $\beta\mapsto\frac{1}{\beta}$,
$$
k_{\mu\nu}=q_{\nu}e^{\beta(u_{\mu\nu}-\frac{\lambda_{\mu}}{p_{\mu}})},
$$ 
and (\ref{eq:lagr2_kern}) gives
\begin{equation}
\sum_{j} k_{\mu j}=1=e^{-\beta\frac{\lambda_{\mu}}{p_{\mu}}}\sum_{\nu}q_{\nu} e^{\beta u_{\mu\nu}}\quad\Leftrightarrow\quad\lambda_{\mu}=\frac{p_{\mu}}{\beta}\ln(\underbrace{\sum_{\nu}q_{\nu} e^{\beta u_{\mu\nu}}}_{=:Z_{\mu}}),
\end{equation}
with $Z_{\mu}$, the $\mu$th-partition function.

Then the following system of self-consistent, $\beta$-dependent equations must hold simultaneously for all $\mu=1,\dots,M$ and $\nu=1,\dots,N$:
\begin{equation}
\label{eq:self_consistent}
\begin{split}
k_{\mu\nu}(\beta)&=q_{\nu}\frac{1}{Z_{\mu}}e^{\beta u_{\mu\nu}},\\
q_{\nu}(\beta)&:=\sum_i p_i k_{i\nu},\\
Z_{\mu}(\beta)&:=\sum_j q_j e^{\beta u_{\mu j}}.
\end{split}
\end{equation}
The system of equations~(\ref{eq:self_consistent}) shows that each row of the solution corresponds to a distinct Boltzmann-Gibbs distribution. The row-independent Gibbs part $q$ is given by the $Y$-marginal of the joint probability $\pi_{\beta}:=\P\rtimes K(\beta)$. By changing the notation slightly, and given that we have identified the optimal value $\beta=\beta^*$, we get
\begin{equation}
\pi_{\beta}(x,y)=p(x)\frac{q_{\beta}(y)}{Z_{\beta}(x)}e^{\beta U(x,y)},
\end{equation}
as the $\beta$-dependent solution.

Also, the system~(\ref{eq:self_consistent}) is equivalent to $M\cdot N$ quadratic equations in $M\cdot N$ variables $k_{\mu\nu}$, with $\beta$ as a free parameter. So, 
\begin{equation}
k_{\mu\nu}\sum_{ij}p_i k_{ij} e^{\beta u_{\mu j}} -e^{\beta u_{\mu\nu}}\sum_i p_i k_{i\nu}=0,
\end{equation}
must hold for all $\mu=1,\dots,M$ and $\nu=1,\dots,N$, which, according to the affine version of B{\'{e}}zout's theorem (cf.~\cite[${\S}$ 7]{cox2015ideals}), has at most $2^{MN}$ complex $\beta$-dependent solutions if counted with multiplicities. 

If $\beta\neq0$, then~(\ref{eq:lagr4_kern}) implies that $R=I(\P;K)$. The optimal value for $\beta^*$ is thus given by
\begin{equation}
\label{eq:rate_beta}
R=\sum_{ij}\frac{p_iq_j}{Z_i}e^{\beta u_{ij}}(\beta u_{ij}-\ln Z_i)=\beta\E_{\pi_{\beta}}[U]-\sum_i p_i\ln Z_i(\beta).
\end{equation}
Then the maximum expected utility~(\ref{eq:constrained_information}), as a function of $\beta$, is given by 
\begin{equation}
\label{eq:utility_beta}
\bar{U}(\beta)=\sum_{ij}u_{ij}p_i\frac{q_i(\beta)}{Z_i(\beta)}e^{\beta u_{ij}}.
\end{equation}
\begin{prop}[Local strict monotonicity]
\label{prop:slope_utility_rate}
The slope of the utility rate function at $(\bar{U}_{\beta},R_{\beta})$, parametrically generated from~(\ref{eq:rate_beta}) and (\ref{eq:utility_beta}), is equal to $1/\beta\in\R^*_+$, i.e. 
\begin{equation}
\frac{d\bar{U}}{dR}(R_{\beta})=\frac{1}{\beta}.
\end{equation}
The utility function $\bar{U}(R)$ is strictly increasing at $R_{\beta}$.
\end{prop}
\begin{proof}
The derivative of the inverse function is given by $dR/d\bar{U}=\beta$, which corresponds to Berger's result for the rate distortion function~\cite[Theorem 2.5.1]{Berger1971}. The claim is now proven by applying the `Inverse Function Rule'. The second statement is a direct consequence of the first.
\end{proof}
\subsection{The utility expansion path}
Let us discuss the geometry underlying the optimisation problem~(\ref{eq:constrained_information}). 
\begin{lem}
Let $U$ be a generic utility matrix. The minimum of the rate-utility function is \begin{equation}
\bar{U}(0)=\E_{\P\otimes \delta_{y^*}}[U]=\sum_{x\in X} p(x)U(x,y^*),
\end{equation}
with $y^*:=\argmax_{y\in Y}\E_{\P\otimes \delta_y}[U]$. It is assumed for $\pi(x,y)=p(x)\cdot\delta_{y^*}$ and for the constant kernel $K_{\delta_{y^*}}$.
The maximum of the rate-utility function is 
\begin{equation}
    \bar{U}(R_{\max})=\sum_{x\in X}p(x) U(x,y^*(x)),
\end{equation}
with $y^*(x):=\argmax_{y\in Y}U(x,y)$. 
It is assumed for 
$\pi(x,y)=p(x)\delta_{yy^*(x)}$ and for the kernel $K(x,y)=\delta_{y^*(x)y}$, where $\delta_{ij}$ is the Kronecker delta. 
\end{lem}

\begin{proof}
For $R=0$, we have $I(\pi)=0$ iff $\pi=\P\otimes\nu$ for some $\nu\in\Delta(Y)$. Given that the maximum is attained at an extreme point of $C_0=\{I(\pi)=0\}$, $\nu$ must be a vertex of $\Delta(Y)$, i.e. a Dirac measure. The specific vertex is determined by the columns of the utility matrix.

For $R=R_{\max}$, the maximising distribution must be an extreme point of the compact convex set $\mathscr{M}_{\P}$, since the mutual information does not impose any further restrictions on the solution space. So the solution is $\pi_{R_{\max}}:=\argmax_{\pi\in\operatorname{extreme}(\mathscr{M}_{\P})}\E_{\pi}[U]$. 

However, depending on the structure of $U$, multiple solutions are possible. Thus, for each choice $j^*(i)\in\argmax_j u_{ij}$, the corresponding stochastic kernel is given by
    $$
    K(x,y)=\delta_{yy^*(x)},
    $$
and with the joint distribution $(\P\rtimes K)(x,y)=p(x)\delta_{yy^*(x)}$. 
Nevertheless, the resulting maximum expected utility is independent of the choices made and is equal to 
    $$
    \bar{U}(R_{\max})=\sum_{x\in X}p(x) U(x,y^*(x)).
    $$ 
\end{proof}
We identify a $M\times N$ utility matrix $U$ with a point in $\R^{MN}$. The hyperplane defined by $U$ is given by $\{x~|~\sum_{ij}u_{ij}x_{ij}=0\}$, with unit normal vector
\begin{equation}
\boldsymbol{n}:=\frac{U}{\| U\|}.
\end{equation}

For a general utility matrix, the associated hyperplane should be the support hyperplane of the mutual information hypersurface. 
The contact point $\pi_R$, where the two hypersurfaces are tangent to each other, is given by 
\begin{equation}
\label{eq:tangent_point}
\boldsymbol{n}=\frac{\nabla I(\pi_R)}{\|\nabla I(\pi_R)\|}.
\end{equation}
\begin{df}
The utility expansion path $\gamma_+:[0,R_{\max}]\rightarrow\mathscr{M}_{\P}$ is defined, for increasing values of $R$, as the set of points satisfying~(\ref{eq:tangent_point}). 
\end{df}
In general, the expansion curve is neither $e$- nor $m$-parallel, cf.~Figure~\ref{fig:Utility_rate}.
If the utility hyperplanes are transverse to the mutual information isohypersurfaces, then the path degenerates to an extreme point corresponding to a Dirac measure. Note that this construction follows the same principle as the construction of the wealth expansion path in consumer choice theory~\cite[p. 24]{Mas-Colell1995}. 

The following simple but nontrivial $2\times 2$ example shows that the expansion path may be constant. 

For $\P:=(p,1-p)$, $p\in(0,1)$, $U: =\left(\begin{array}{cc}1 & 0 \\ 1 & 0\end{array}\right)$, $K=\left(\begin{array}{cc}k_{11} & k_{12} \\ k_{21} & k_{22}\end{array}\right)$ and $\pi:=\P\rtimes K$, the expected utility w.r.t $\pi$ is $\E_{\pi}[U]=pk_{11}+(1-p)k_{21}$. It is maximal for $k_{11}=k_{21}=1$, and therefore the extreme stochastic matrix is equal to $K_{\max}=\left(\begin{array}{cc}1 & 0 \\1 & 0\end{array}\right)$. The corresponding $Y$-marginal is $\hat{\pi}_Y=(1,0)$, i.e. the Dirac measure $\delta_0$ supported at the vertex $0$ of $[0,1]$.
The maximum of the expected utility over $\mathscr{M}_{\P}$ is equal to $1$, and it is assumed for the product measure $\pi_{\max}=\P\otimes\delta_0$, with $I(\pi_{\max})=0$. In contrast, for the joint distribution $\pi:=(p,0,0,1-p)$, we have $I(\pi)=-(p\ln p+(1-p)\ln(1-p))>0$ for $p\in(0,1)$, which is exactly the entropy of the distribution $\P$.

The dual objective of utility minimisation is given by
\begin{equation}
\label{eq:utility_rate_min}
\bar{U}(R):=\min_{\substack{\pi\in\mathscr{M}_{\P} \\ I(\pi)\leq R}}\E_{\pi}[U].
\end{equation}
From a geometrical point of view, this corresponds to the solution of the equation with the opposite sign 
\begin{equation}
\label{eq:utility_contraction_path}
-\boldsymbol{n}=\frac{\nabla I(\pi_R)}{\|\nabla I(\pi_R)\|}.
\end{equation}
\begin{df}
The utility contraction path $\gamma_-:[0,R_{\max}]\rightarrow\mathscr{M}_{\P}$ is defined, for increasing values of $R$, as the set of points satisfying~(\ref{eq:utility_contraction_path}). 
\end{df}
It can be shown that $\gamma_+$ and $\gamma_-$ are disjoint and related by a reflection through a point (involutive isometry), since for any fixed value of the mutual information the hypersurface consists of two disjoint components, cf.~Figure~\ref{fig:Utility_rate}.

\begin{figure}
    \centering
    \includegraphics[width=1\linewidth]{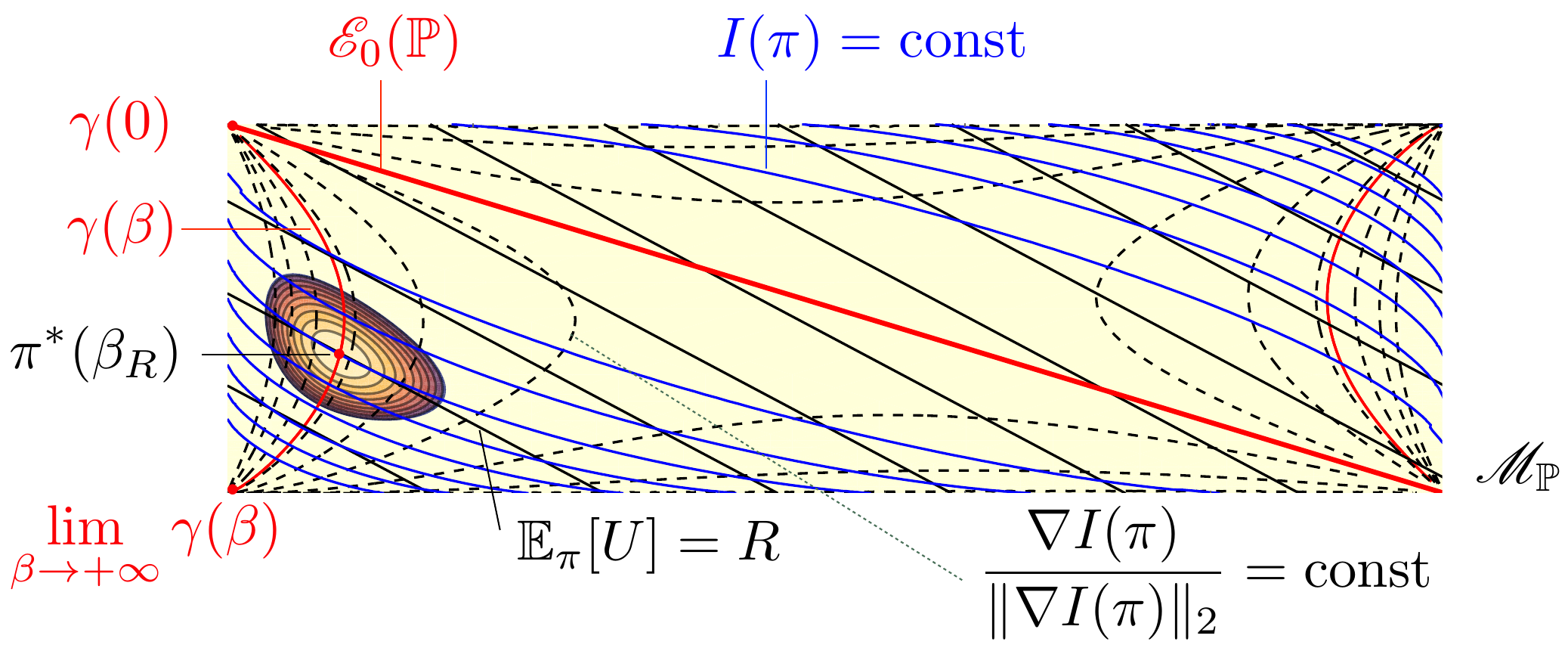}
    \caption{{\color{red}Red diagonal line} {\color{red}$\mathscr{E}_0(\P)$}: Locus of product measures in $\mathscr{M}_{\P}$. The dashed contours show mutual information values with $I(P;K)=R$ for different values of $R$. The {\color{blue}blue contour lines} show values of $\E_{P\rtimes K}[U]$ for $\bar{U}(R)\geq R_0$. The {\color{red}red curved lines} show the paths of utility expansion ${\color{red}\gamma_{\max}(\beta)}$ and contraction ${\color{red}\gamma_{\min}(\beta)}$.}
    \label{fig:Utility_rate}
\end{figure}
\section{An example of application to legal theory}
\label{sec:legal_app}

Let $\mathfrak{G}_x$ be the set of exercisable (fundamental) rights depending on a subject (personal scope) or a situation (material scope) or both, represented by a vector $x$. 
Let $\mathfrak{C}_x\subset\mathfrak{G}_x$ be the subset of fundamental rights that cannot be restricted in principle. 

For $R$, an $n$-ary relation, and $\varphi$ a propositional formula, let $\sigma_{\varphi}(R)$ be the function that selects (restricts to) from $R$ those $n$-tuples for which the formula $\varphi$ holds.

Applying the selection operation yields $\sigma_{\varphi}(\mathfrak{O}_x)\subset\mathfrak{O}_x$, which must satisfy the following conditions: The conditions that make up $\varphi$ must satisfy constitutional or legal requirements, e.g. they must be non-discriminatory, and $\mathfrak{C}_x\cap\sigma_{\varphi}(\mathfrak{O}_x)=\emptyset$. Let $\hat{a}\in\sigma_{\varphi}(\mathfrak{O}_x)$ denote a right (element) that has been restricted.

Let $\mathfrak{L}\subset X$ be the subset of legal states of the world. For a restriction to be justified, $\hat{a}.x\in\mathfrak{L}$ and $\hat{a}.x\succsim x$ must hold, i.e. the restriction must lead to a legal state that is at least as desirable as the current state.  

Let $U_x(\hat{a}):=U(x,\hat{a})\in\R_+$ (as in~(\ref{eq:utility})) be the net public utility,\footnote{A cost-benefit analysis must be conducted and an aggregate utility function separately modelled according to the principles of welfare economics. However, this approach presents its own fundamental challenges, cf.~\cite[Chapter~22]{Mas-Colell1995}.}  the sovereign achieves for the `greater good' by restricting $\hat{a}$ in the context of $x$. Let $\kappa_x\in\Delta^{\circ}(\mathfrak{O}_x)$ be the a priori probability of $x$ choosing from the set of unrestricted rights, and $p\in\Delta(\sigma_{\varphi}(\mathfrak{O}_x))$ be the probability of the sovereign choosing from the set of restricted rights.  

Let $d$ be a metric or a divergence between probability measures. A disutility function $D$ (or $D_x$ in a fixed state $x$) is a strictly decreasing convex function of $d$ with values in $\overline{\R}_+:=\R_+\cup\{+\infty\}$, such that $D(0)=D_{\max}>0$ and $D(d)\rightarrow0$ for $d\to+\infty$. The three basic types are (not to scale)
\begin{equation*}
d\mapsto\frac{1}{d}: \includegraphics[width=0.19\linewidth]{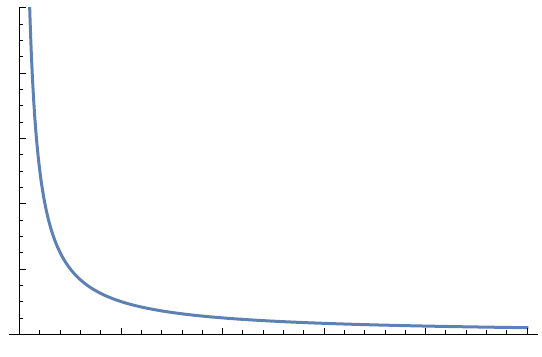},\quad d\mapsto e^{-d}: \includegraphics[width=0.19\linewidth]{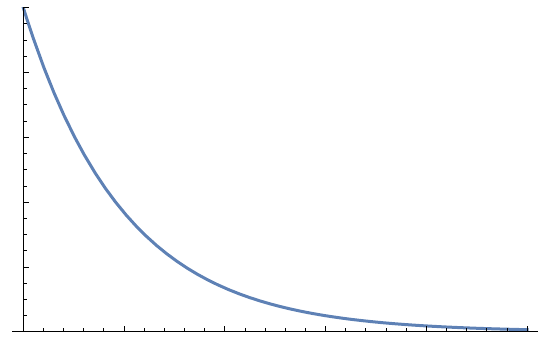},
\quad d\mapsto D_{\max}-d: \includegraphics[width=0.19\linewidth]{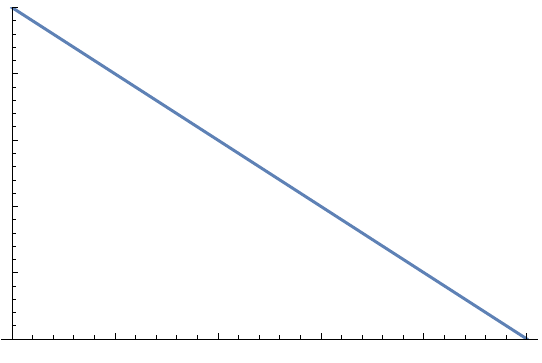}
\end{equation*}

The optimisation problem for the end-means relation, called proportionality in the strict sense, is, for both $\beta\in\R_+$ and $\kappa_x\in\Delta^{\circ}(\mathfrak{O}_x)$ fixed, given by
\begin{eqnarray}
\label{eq:subsidiarity1}
p^*_{\beta}&:=&\underset{p\in\Delta(\sigma_{\varphi}(\mathfrak{O}_x))}{\mathrm{argmax}}\left(\E_p[U_x]-\frac{1}{\beta}D_x(d(p,\kappa_x))\right),\\\nonumber
&\text{s.t.}& \beta\cdot\E_{p^*_{\beta}}[U_x]>D_x(d(p^*_{\beta},\kappa_x)).
\end{eqnarray}
So if the net benefit to the public is positive, the sovereign has a right to act.

Example: For $m<n$, let $\Delta_m\subset\Delta_n$ be a $m$-dimensional face and $q\in\Delta^{\circ}_n$ be a fixed measure with full support. 
Then the maximum divergence w.r.t. $q$ is equal to $d_{\max}:=d_{\max}(q):=-\ln q_{i_{\min}}$, $i_{\min}:=\argmin_{i\in[n]}q_i$ with $i_{\min}\notin[m]$, as we assume. 

Define $D(d(p,q)):=d_{\max}-D_{\operatorname{KL}}(p||q)$ for $p\in\Delta_{m}$ as the disutility function. Let the sovereign's utility vector (a row of the utility matrix in situation $x$) be $U=(u_0,u_1,\dots,u_m)$.
Then the net expected utility $F_{\beta}[p]:=\E_p[U]-\frac{1}{\beta}D(d(p,q))$ is a continuous convex functional (as the sum of two convex functions; recall that for fixed $q$, $D(\cdot,q)$ is convex, cf.~\cite[Theorem 2.7.1]{cover1991elements}) which has a maximum at an extreme point of $\Delta_m$, i.e. a vertex $j\in[m]$. 
Put $d^*_{j^*}:=d(\delta_{j^*},q)$ where $j^*:=\argmax_{j\in[m]}d(\delta_j,q)$ and set $\xi^*:=\argmax_{\xi\in[m]}u_{\xi}$ with $u_{\xi^*}$ the corresponding maximum value.
The minimum divergence $d_{\min}>0$ between $q$ and $\Delta_m$ is equal to the conditional probability
$$
q|_{\Delta_m}=\frac{1}{\sum_{j=0}^mq_{i_j}}(q_{i_0},q_{i_1},\dots,q_{i_m}).
$$
\begin{minipage}{0.4\textwidth}
Then (cf.~Figure~\ref{img:bound_disut}) $d(\cdot,q)$ is a continuous bounded functional on the face $\Delta_m$ with 
$$
0<d_{\min}\leq d(p,q)\leq d^*_{j^*}<d_{\max}.
$$
For $\beta\to+\infty$, $p^*=\delta_{\xi^*}$ and $F[p^*]=u_{\xi^*}$ is the maximum public utility (note that $\xi^*$ and $j^*$ are not necessarily the same). 
\end{minipage}%
\begin{minipage}{0.5\textwidth}
\begin{center}
    \includegraphics[width=0.7\textwidth]{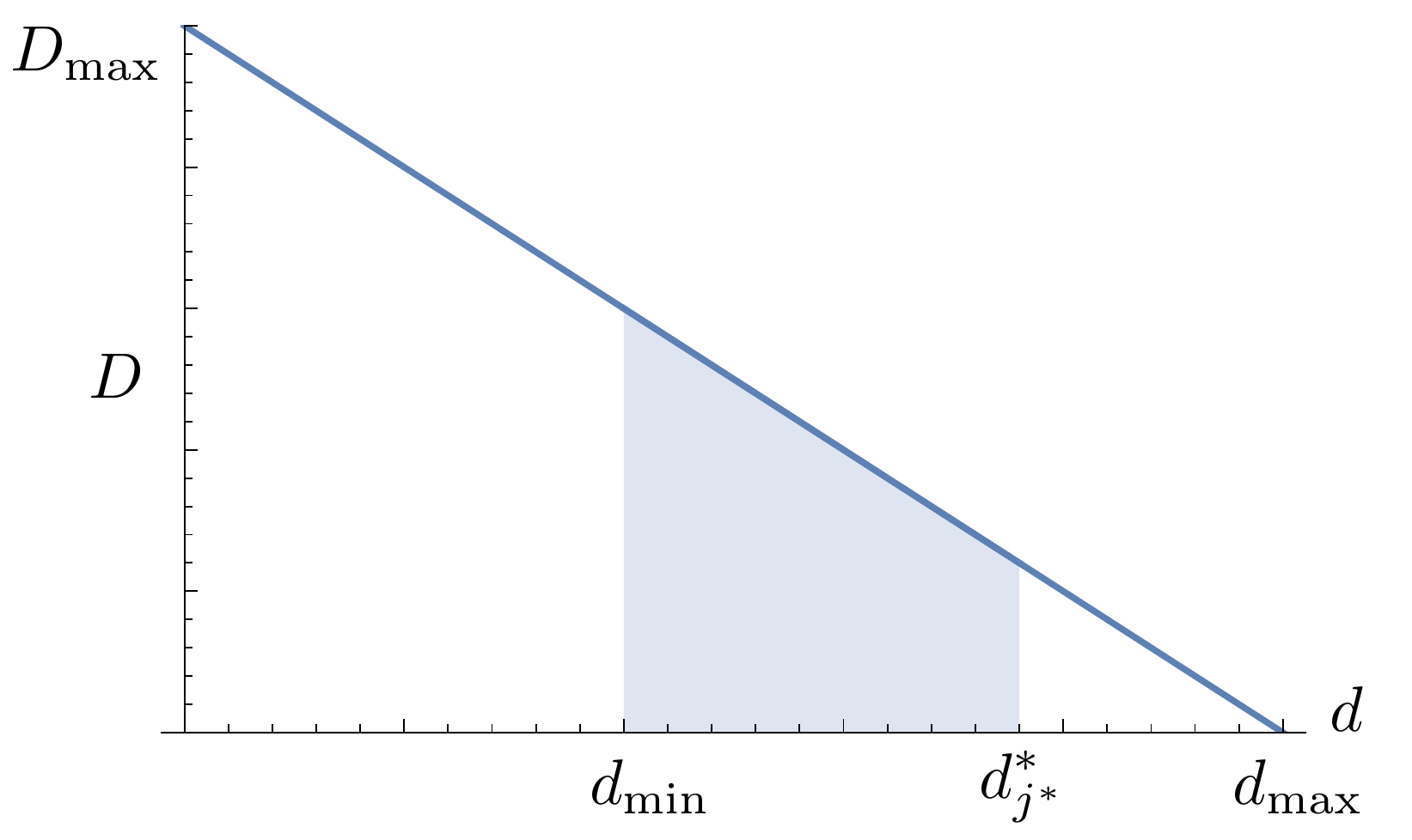}
    \captionof{figure}{Disutility function $D(d)$}
    \label{img:bound_disut}
\end{center}
\end{minipage}

For $\beta\to0$, $F_{\beta}[p]<0$ for all $p\in\Delta_m$, because $D(d)\geq d_{\max}-d_{j^*}>0$, i.e. it is bounded from below, and thus~(\ref{eq:subsidiarity1}) has no solution.
Between the two extremes there exists $\beta_0>0$ such that $F_{\beta_0+\varepsilon}[p]>0$ for some $p\in\Delta_m$ and for all $\varepsilon>0$. 
Furthermore, for a $\beta_i>\beta_0$, $i\in\{1,\dots,\ell\}$, the solution can change discontinuously from vertex $j_i\to j_{i-1}$, $j_0:=\xi^*$, such that $p^*_{\beta_i}=\delta_{j_i}\to p^*_{\beta+\varepsilon}=\delta_{j_{i-1}}$, cf. Figure~\ref{fig:Disutility_surface}.
Unlike in ~(\ref{eq:constrained_opti}), we have a second best problem. Since $D(d)$ is bounded, it is not possible to compensate for increasing values of the coupling constant by moving far enough away from the prior. 
\begin{figure}
    \centering
    \includegraphics[width=0.6\linewidth]{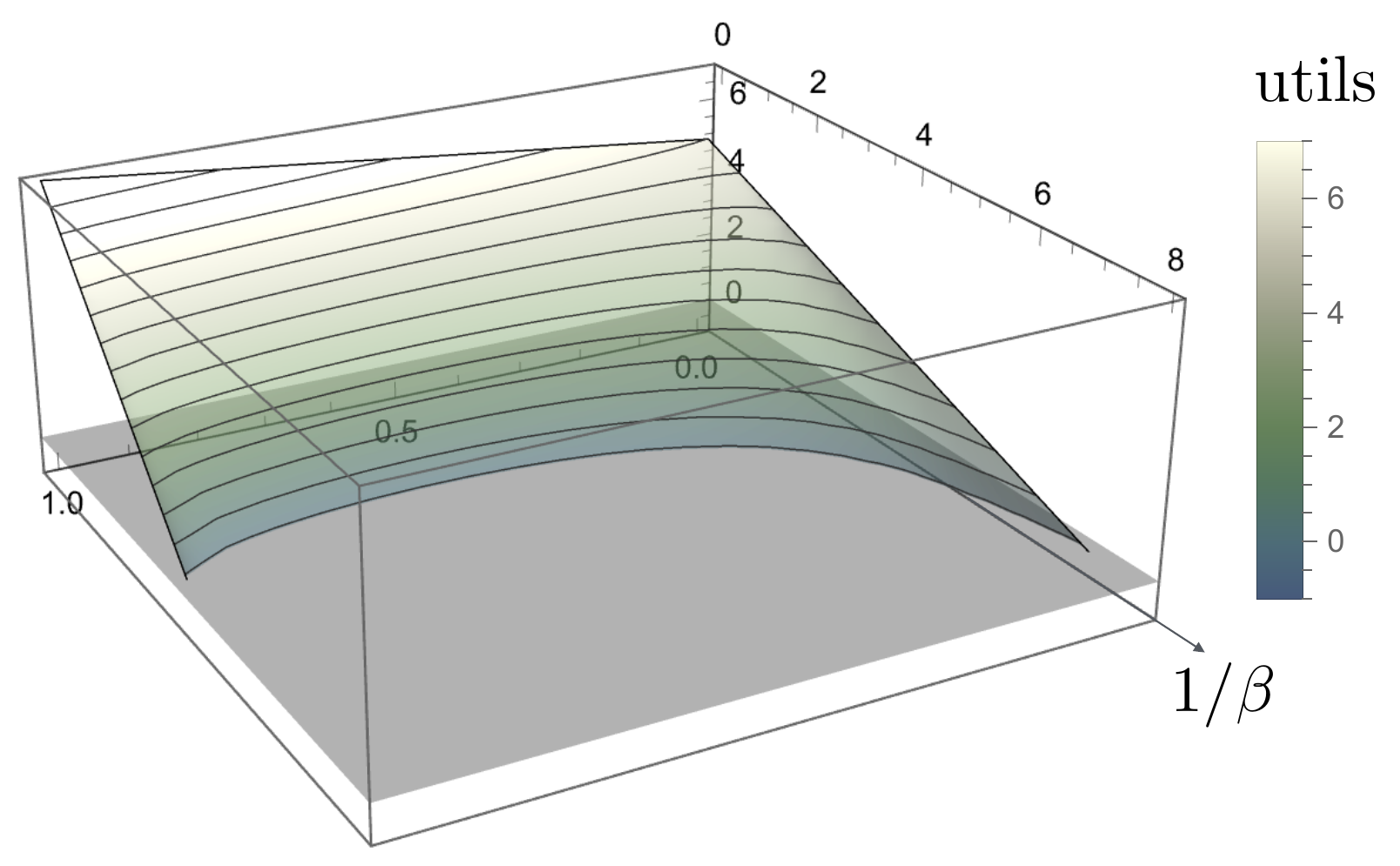}
    \caption{The $z$-axis shows the values of the net expected utility $F_{\beta}[p]$ as a function of temperature $1/\beta\in[0,8]$, probability $p\in\Delta_1$, prior $q=(0.7,0.2,0.1)\in\Delta^{\circ}_2$ and utility matrix $U=(7,5)$. For a critical temperature $1/\beta_0\in[0,8]$ the optimal solution $p^*_{\beta}$ changes discontinuously from the Dirac measure $\delta_0$ to $\delta_1$.} 
 \label{fig:Disutility_surface}
\end{figure}

So, let us summarise our findings.
Our examination of the information geometry underlying bounded rational decision-making and its implications for the field of jurisprudence has yielded a number of remarkable insights. 
In essence, it can be succinctly summarised by analysing the components of the main object of our investigation, which is the following type of formula
\begin{equation}
\label{eq:fundamental_eq}
\max_{p\in\Delta(\operatorname{supp}(q))}\left\{\underbrace{\E_p[U]}_{\text{utilitarian}}-\frac{1}{\beta}\underbrace{\mathcal{R}(p,q)}_{\text{deontic}}\right\},
\end{equation}
where $\mathcal{R}$ is a regularisation function and $q$ is a Bayesian prior.
Expected utility is a statistical utilitarian approach to improving the global situation. However, it requires the formulation of a (fair) aggregate utility function, which is the domain of social choice and welfare economics~\cite{Mas-Colell1995}.
In contrast, $\mathcal{R}$, with the help of the prior $q$ and its support, serves as a cut-off and thus incorporates deontological considerations for the protection of individual rights, which expected utility cannot provide. 
If reality provides a source probability for the states of the world and replaces a fixed Bayesian prior, the range over which optimisation can be performed is still dictated by the allowed support for the measures representing the policies to be considered. 
However, neither utilitarian nor deontological considerations can determine the coupling constant itself. It is therefore ultimately up to a third party, such as the legislator, to determine an acceptable range of values, i.e. the margin of discretion.

\subsection*{Acknowledgments}
I would like to thank Takuya Murayama (Kyushu University) for his help in proving the strict monotonicity of the free energy, his suggestion to use regular conditional probabilities, and for general comments. 
\bibliography{references_MG.bib}
\bibliographystyle{acm}
\end{document}